\documentclass{llncs}
\usepackage{epsfig}
\usepackage{amssymb,amsmath}
\usepackage{mathrsfs}
\usepackage{graphicx}
\usepackage{multirow}
\usepackage{url}
\usepackage{color}

\usepackage{booktabs,xcolor,siunitx}
\definecolor{lightgray}{gray}{0.9}
\usepackage{wrapfig,framed,caption}
\usepackage{graphicx}

\graphicspath{ {images/} }

\mathchardef\myhyphen="2D

\newcommand{\remove}[1]{}

\newcommand{\bI}{{\mathsf I}}
\newcommand{\bH}{{\mathsf H}}
\newcommand{\bPr}{{\mathsf P}}

\newcommand{\cM}{{\cal M}}
\newcommand{\sfM}{{\mathsf M}}

\newcommand{\bF}{{\mathbb F}}
\newcommand{\SD}{\mathsf{SD}}

\newcommand{\bl}{\color{blue}}

\newcommand{\ext}{\mathsf{Ext}}

\newcommand{\rwc}{{$(\rho_r, \rho_w)$-channel}}

 \newcommand{\rw}{{$(\rho_r, \rho_w)$}}

\newcommand{\cK}{{\cal K}}
\newcommand{\cO}{{\cal O}}

\newcommand{\sfH}{\mathsf{H}}

\newcommand{\sfI}{\mathsf{I}}

\newcommand{\rs}{\mathsf{RS}}

\newcommand{\Adv}{\mathsf{Adv}}

\newcommand{\LVACenc}{\mathsf{LVACenc}}
\newcommand{\LVACdec}{\mathsf{LVACdec}} 
\newcommand{\MAC}{\mathsf{MAC}}

\newcommand{\sk}{\mathsf{SK}}
\newcommand{\sfC}{\mathsf{C}}

\newcommand{\lv}{\mathsf{LV}}

\newcommand{\bequation}{\begin{equation}}
\newcommand{\eequation}{\end{equation}}

\begin{document}
\title{Information-theoretically Secure Key Agreement over Partially Corrupted Channels}

\author{Reihaneh Safavi-Naini and Pengwei Wang\\
Department of Computer Science, University of Calgary, Canada\\
reisaf@ymail.com  and pengwwan@ucalgary.ca}
\maketitle

\begin{abstract}

Key agreement is a fundamental cryptographic primitive. 
It has been proved that key agreement protocols
with security against computationally unbounded adversaries cannot exist in a 
setting where Alice and Bob do not have  dependent variables and 
communication between them is fully  public, or fully controlled by the adversary. 
In this paper we consider this  problem when the  adversary
can ``partially" control the  channel.   We motivate these adversaries by considering
adversarial corruptions at the physical layer of communication,  give a definition
of adversaries that can   ``partially"  eavesdrop and  ``partially"  corrupt the communication.
We formalize 
security and reliability of key agreement protocols, derive bounds on the rate of
 key agreement, and give constructions that achieve the bound. 
 Our results show that
it is possible to have secret key agreement as long as some of the communicated 
symbols remain private and unchanged by the adversary. 
 We relate our results to the previous known results, and
discuss future work.

\remove{
udied
in two communication settings, first when the communication is over a public
authenticated channel, and second when the channel is fully controlled 
by the adversary,  The protocol may start either without any prior shared variables, or when 
variables $X$, $Y$ and $Z$, with probability distribution $P(X, Y, Z)$ 
is shared by Alice, Bob and Eve respectively.  It is proved that such protocols
impossible  without a shared
variable, and in presence of an active unbounded adversaries.
}
\end{abstract}


\section{Introduction}

One of the fundamental problems in cryptography is establishing a shared secret key between two parties.
% who do not have a  prior shared secret.  
 The problem has been  studied in different settings. % under different assumptions.
%in different settings, 
Important distinctions among settings are based on,  (i)  the adversary's computational power,
% is  polynomially bounded, or the adversary is computationally unlimited, 
(ii)  if %the communication channel connecting the two participant is reliable against active 
communicants have access to a public authenticated channel,
 or  if the  communication is over tamperable channel,   and, (iii) if there are initial shared information, possibly in the
 form of dependent random  variables.
\remove{ There are also other assumptions such as existence of hardware tokens \cite{GISV10}, public randomness source (random beacon) \cite{??} or common reference string \cite{??}.

   Diffie-Hellman ground breaking protocol  \cite{DH76} correspond to the setting of computationally bounded adversary, public authenticated channel and no correlated randomness. Extensions of this protocol \cite{} also allow the channel to be tamperable. 
}

We consider key agreement  with security against computationally  unlimited adversaries.  Information theoretic key agreement  was first considered by Maurer \cite{M92,MAU93} with the motivation of
%  Maurer - Secret key agreement by public discussion from common infromation  - Transaction on IT 93
providing  positive results  for scenarios that secure communication, with security against a computationally unlimited adversary, had been
proved to be impossible.
The two main approaches to securely sending a message over a channel that is eavesdropped by an unlimited adversary are due to 
Shannon %sec the approaches of   Shannon's model of perfect secrecy 
\cite{S48} who formalized the notion of perfect secrecy when Alice and Bob are connected by a public reliable channel, 
and second,  Wyner \cite{W75}  who introduced wiretap model  %, give negative results.  In Shannon's setting, allowing the adversary to have perfect  (eavesdropping) access to communication channel results in the requirement of a shared key the same length as the message and so without such prior key, secure communication is impossible. In Wyner's setting  there
in which communicants use the noise in the channel to provide  prefect secrecy for the communication 
without requiring  a shared secret key.
Security guarantee in both these models  although strong, is only achievable under very restrictive assumptions.
Perfect secrecy in Shannon's model requires communicants to  share a secret key of the length  at least the size of the message.
%n wbut assumes that the adversary does not have a perfect view of the
%communicated channel.
Secure communication in wiretap setting  is only possible if
%Alice sends  encoded messages to Bob, that because of the noise in the
%Eve's channel becomes partially visible to Eve.
 %Wyner proved  perfectly secure message transmission is possible only if
 Eve's view of the codeword is ``noisier" than the Bob's view of the codeword,
which does not hold  for settings that the eavesdropper is closer to Alice  than Bob,
  and has a better communication channel to Alice.
 %secure communication will be impossible.
Thus positive result in both above  approaches are under conditions that are of limited practical applications.

Motivated by this observation,
\remove{
To provide positive results for % secure communication in
 a wider range of scenarios, 
 }
 Maurer 
considered  the more basic problem of secret key agreement (secure communication is
possible, if one can have a shared key)
%formalized and studied a model 
 where Alice and Bob want to share  a secret key while the communication channel between them is
 eavesdropped by a computationally  unlimited adversary.
  Maurer considered a minimum setting  where Alice, Bob and Eve hold dependent variables $X$, $Y$ and $Z$, %respectively, such that
with a joint distribution $\bPr_{XYZ}$,
and Alice and Bob can interact over  a {\em public discussion (PD)} channel:  an authenticated channel that is fully
visible by all  system participants.  
The setting is ``minimum" in the sense that it was  proved \cite{MAU93} that without any initial joint distribution, 
secure  key agreement is impossible.
The joint partially leaked distribution  $\bPr_{XYZ}$ is the only resource of Alice and Bob, and
so a basic question is when a secure shared key can be derived by interacting 
over the PD.
The joint distribution can be generated  by  receiving transmission of a public random beacon (e.g broadcast by a satellite) that 
broadcasts samples of a random variable, and this is received by all parties 
%by receiving a random sequence generated by a public beacon, 
over their individual channels. %  associated with each party in the system.
The distribution can also be simply generated by  Alice sending a random string $X^n$ of length $n$ to  Bob,
 over  a  discrete memoryless wiretapped 
channel,  resulting in $Y^n$ and $Z^n$, at Bob and Eve respectively.
%physical layer noisy channel.

Maurer showed that interaction over PD in the above setting is indeed powerful,  and could result in 
a shared random key and hence secure communication,  when  Shannon and Wyner
 models give negative results.   
\remove{
 and with this initialized distribution $\bPr_{XYZ}$, assuming an interactive public discussion channel, Alice and Bob can generated shared keys, significantly by expanding  the range of settings
that privacy against an unbounded eavesdropping adversary is required.
}

Maurer \cite{MAU97}  later considered the case that %also  initiated the study of  key agreement over
 no PD is available, and   communications are over a tamperable channel. % ``NOT authenticated channel".  
The adversary in this setting can    completely control the communication and % by the Eve who can
modify or drop the sent  messages. 
Intuitively when the adversary is  computationally unbounded and fully controls the 
 communication channel,  one cannot expect any secure key to be established. %communication.
Maurer proved that without initial correlated variables, this is indeed the case even if extra 
restrictions (e.g. secure channel in one direction) are placed on the adversary.
He also proved  existential results for the case that $\bPr_{XYZ}$ exists.

The works \cite{MW03,RW04,KR09,DW09} % consider computationally unbounded adversaries and  unauthenticated channel, 
all consider this powerful model of adversary, corresponding to  {\em Dovel-Yao} model of adversary that is applicable  in networks
\cite{DY83}.
To enable establishment of shared key however, authors  assume strongly correlated variables, in particular identical  \cite{RW04} or   ``close" secrets  \cite{KR09,DW09},
 that are partially leaked.
The role of 
interaction in these settings is ``reconciliation" that results in a shared secret (in the case of close secrets), 
and privacy amplification that   \cite{RW04,BBCM95}  %the interaction does not reduce the entropy of the correlated variable,
%and the final extracted key has most of the entropy of the initial variables.
extracts the randomness in the shared  variable.

\remove{
A related problem  is  two-party information-theoretic key agreement  when the two parties  have ``close" secrets that is partially leaked to the adversary: the two close partially leaked secrets correspond to $X$ and $Y$, and $Z$ is the adversary's information about the secrets. The problem has also been called information reconciliation and privacy amplification \cite{RW04}. In this setting, existence of the initial correlated variables (close secrets) effectively reduces the key agreement problem to error correction and extraction over adversarially controlled channel.
}

In this paper we consider a setting where no dependent  variables is held by Alice and Bob.
Maurer showed impossibility of  key agreement when there is no shared correlated variables, %without prior shared key 
 and the adversary can either completely eavesdrop the communication channel, or
fully control the  channel.  
We ask  if positive results could be obtained when 
%reduces corruption power of 
the adversary ``partially" sees and ``partially" controls the channel.

\remove{
there are  ``reasonable"   adversarial settings, in which secure key agreement 
becomes possible. Obviously this question could only have positive answer if the adversary does not fully control the channel
and  we need to restrict the adversary's capabilities, with the goal of capturing 
%We however would be interested in adversaries that model 
realistic scenarios.
}

Adversaries  with {\em partial}  access to the channel can be naturally defined
% instead  of considering adversaries 
%at the  higher
at the lowest layer  of OSI  (Open Systems Interconnection) protocol stack, known as the {\em physical layer}.
\remove{
 (i.e. transport layer, IP layer) that are able to replace or drop
the whole 
information packets, 
we consider adversaries 
 that operate at the %lower layer of protocol stack,  known as
{\em physical layer} of OSI protocol stack.
}
This is the same OSI layer that is used in the Wyner wiretap model where
partial view of the adversary (due to noise in Eve's channel)  is used to provide secure communication. %l as a resource.
%to model eavesdropping. 
Here we  consider adversaries with partial view of the message and partial  tampering  ability.
%at this layer also.

Information units   at  the physical layer of communication correspond to elements of a $q$-ary alphabet \cite{Eli57,GR06,GS98,HAMMING50,S48}.
The goal of encoding at this layer traditionally  is to provide reliability %raditionally been used to encode information
 against channel noise.
We consider the case that these units (channel symbols) % encoding of encode wand these are 
can be individually accessed,  changed,  or blocked by a jamming adversary.
In practice, an adversary  with 
a transceiver, depending on their  location and transceiver capability, can  intercept some  of the transmiited symbols, 
and/or add adversarial noise to corrupt   them. 
An adversary with full control of all communicated symbols corresponds to the network layer adversary in \cite{MAU97} and 
the follow up works (same as Dolev-Yao model). %, with negative results in the case that no initial jointly distributed variables exist.
By moving to physical layer, we  are able to consider adversaries with different levels of control over the 
communicated messages,  and study the key agreement problem against a more refined  classes  of adversaries,
that capture corruption at physical layers of communication.

\subsection{Our work}

%, while less emphasis is put on cryptographic definitions of adversary and  security notions.
%theoretical modelling and analysis, are limited to the network layer adversary.
We initiate cryptographic study of key agreement problem 
in presence of physical layer adversaries, and show  positive results  %can be obtained
when  key agreement in presence of network layer full eavesdropping and/or corrupting adversaries,  is impossible.

\remove{
In particular,  secret key agreement problem focusses on ``intuitively" secure system 
in presence of a 
 study of information-theoretic secure key agreement when the channel is adversarial, 
has been primarily assuming  network layer tampering adversaries, and using the work of \cite{MAU97,RW04} as
the theoretical  foundation.
}

We consider key agreement problem between Alice and Bob who are 
connected by a channel 
% that there is no a priori  correlated variable held by communicants, and
%assuming a  physical layer  adversary
 that is partially controlled by Eve, and the partially controlled communication is 
 their  % assume  this is the 
 only resource. % of the communicants.
 % controls the  communication channel,
% and this is their only resource.  This is 
\remove{
In wiretap model  partial (noisy) view of the sent codeword is the only resource of the
communicants for secure communication;  in our model partial control of Eve in the
communicants' resource for establishing shared key.
}

 Alice and Bob send   messages back and forth, over the channel.
\remove{
each transmission % round the adversary can 
``partially" controlled by the the adversary. %channel.
}
We define the partial control  of  Eve by their ability to select, (i)
a subset $S_r$ of the transmitted components for eavesdropping, and (ii) a subset $S_w$ of transmitted components, to corrupt
by adding  (jamming) a noise vector. 
 The two sets are chosen adaptively in each round and may have overlap.
 We impose 
the restriction that in each round, $|S_r|\leq \rho_r n$ and $|S_w|\leq \rho_w n$,
where $n$ is the length of transcript in that round, and $\rho_r$ and  $\rho_w$ are fixed constants in the range $[0,1]$, specifying
the adversaries capabilities.  A network layer adversary corresponds to  $\rho_r=\rho_w=1$, and a perfectly secure 
authenticated channel corresponds to  $\rho_r=\rho_w=0$.

%The model captures wireless adversaries,   and   
Parameters $\rho_r$ and  $\rho_w$ model  wireless adversaries'
 limitation of receiving antenna and receiver, and jamming capability, respectively.
As experimentally shown in \cite{PTDC11},  because of the constraints of real systems and channels, 
making a deterministic change to an individual transmitted symbol is ``hard", if possible at all. 
Assuming additive noise captures the uncertain effect of corruption on the transmitted symbol. 
We do not consider an  adversary with unlimited jamming power; such adversaries
 can always completely disrupt the communication.
Our adversary has a fixed eavesdropping and corruption budget in each round.
A stronger adversary  is when the adversary has a total budget for the whole protocol,
and can plan how to spend it in different round with the restriction that,
$\sum_i | S_r^i | \leq \rho_r n$ and  $\sum_i | S_w^i | \leq \rho_w n$. 
We start with  the less demanding case that the budget of each round is fixed. This is also the
more realistic case as the partial view of the adversary in most cases is due to limitations of the adversary's 
hardware and processing capabilities.

We also consider a third parameter, $\rho = \frac{ |S_r \cup S_w| }{n}$ which is the fraction of transmitted
symbols that are either leaked or corrupted.
%and so intuitively cannot contribute to the final shared key.

\vspace{1mm}

\noindent
{\em Models and definition.}
We define a $(\rho_r, \rho_w)$-channel, and use the definition of  (interactive) key agreement protocol
in \cite{M92,MAU97}, replacing the PD with $(\rho_r, \rho_w)$-channels.
The protocol proceeds in rounds,  in each Alice or Bob sends a message. 
At  the end of the protocol, Alice and Bob  output keys $K_A$ and $K_B$, respectively.
The  security 
properties of  a secure key agreement protocol
 are as follows.

\begin{itemize}
\item {\em Strong reliability:}  The probability that Alice and Bob do not derive the same key satisfy, $\bPr (K_A\neq K_B)\leq \delta$; 
\item {\em Security:}  The generated key is private, given Eve's view of the communication; 
\item {\em Randomness.} The  generated key is statistically close to uniform distribution.
\end{itemize}

Defining reliability in information theoretic key agreement, when adversary tampers with the communication, is subtle.
 Maurer  \cite{MAU97}  and the follow up work \cite{KR09}, 
consider the case that at the end of the  protocol, Alice and Bob either   output a key, %element of the set
%$\{ K_A, \perp\}$ ( $\{ K_B, \perp\}$)
or $\perp$. That is they either output a key, or declare the protocol 
unsuccessful.
The protocol is successful when 
% successful as long as no shared leaked key is obtained. That is, the protocol is considered
%successful, if 
at least one of the communicants output $\perp$ (and so guaranteeing that a shared leaked key will not be
established) or, a secret shared key is established. 
 Using  this  definition, a protocol may have negligible failure probability but at that same time in
most cases no {\em shared secret key} be established.
This is a natural definition considering a network layer adversary that completely controls the communication.

Our definition of  {\em strong reliability}  is the same as in \cite{MAU93,MAU97} where 
 communication is over PD and the adversary is passive.
In such  setting (initial $\bPr_{XYZ}$ and communication over PD) Alice and Bob may output 
different keys % not because of corruption by the channel adversary, but
 because of the working of the
protocol and  properties of the functions that are used for the  derivation of the key, and
not because of corruption by the channel adversary. 

A surprising result of our work is that by slightly reducing  the control of the adversary on the channel ($\rho_r$ and $\rho_w$ can be close to 1),
% to be
%quantifiable, 
one can expect strong reliability in presence of corrupted communication.
% and  considering  the protocol to fail,  only in  cases that the key of Alice and Bob are not the same.

Definition of secrecy
 and key randomness is by requiring  that
 the distribution of secret key given Eve's observation $K|Z$, is statistically close  to a  uniformly distributed variable $U$ which is of 
 the same length as $K$.  

\vspace{1mm}
\noindent
{\em Rate and impossibility results.}
Following \cite{MAU97},  we define 
the secret key rate of a protocol as the rate that Alice and Bob agree on a shared key, while Eve's total  information 
remains bounded by $\epsilon$. The rate is 
given by
%We define the rate of a key agreement protocol by
 $R = \frac{\log |\cK|}{n \log |\Sigma|}$, where $\Sigma$  is the channel  alphabet  (symbols sent over the channel).

 We prove,
\[ 
R\leq 1-\rho.
\]

The bound effectively shows that the fraction of symbols that are either eavesdropped or tampered with in each round, cannot
contribute to the secret key. This is intuitively expected and the proof shows that interaction cannot 
overcome this restriction.

In this definition of rate, all communicated symbols  contribute to the communication cost.
In the definition of rate in \cite{M92,MAU97} however, communication over public channel is  free and 
the rate only considers  the number of shared triplets $(X,Y,Z)$ used by the protocol.
When $\bPr_{X^nY^nZ^n}$  is the result of Alice sending $X^n$ to Bob over a wiretap channel,
only the cost of this communication  (over physical 
layer)  is considered in the calculation of rate.  The communication over the PD  is for reconciliation and 
extraction, and are considered free. 

The bound  shows impossibility of secure  key agreement  when $\rho>1$.

The above bound is derived  for  key agreement protocols assuming  strong reliability.
The bound also holds for   AWTP-PD channel under the same condition. 

In Section \ref{SecRate} we derive the  upper bound 
\[ R\leq 1-\rho_r\]
 on the secret key rate of key agreement protocols under weak reliability condition.
The same bound holds for key agreement over AWTP-PD channels also.

%In Section \ref{SecRate}, we show the rate of key agreement over AWTP-PD channel is $R \leq 1 - \rho$. This is same as the upper bound on rate of key agreement protocol over PD channel \cite{MAU97}, since $R \leq \max(\bI(X^n, Y^n), \bI(X^n, Y^n | Z^n)) = 1 - \rho$. 

Maurer's  bound  $R \leq \max(\bI(X^n, Y^n), \bI(X^n, Y^n | Z^n))  $ \cite{MAU97}, can be written as  $R \leq 1 - \rho$ 
because $ \max(\bI(X^n, Y^n), \bI(X^n, Y^n | Z^n)) =1 - \rho $.    
%This implies that the rate of key agreement protocol over AWTP-PD channel is same as
Note that the rate in Maurer's  bound does not consider communication over PD, while in our setting the rate includes
all the communication.

%For key agreement protocol with weak reliability over AWTP-PD channel, the key agreement protocol can have a higher rate 
Using weak reliability, the bound on the secret key rate  is $R_\sk \leq 1 - \rho_r$. 
 Note that  in our setting the shared dependent variables  $X,Y,Z$ are influenced by Eve through 
the choices of $S_r$ and $S_w$ and the added noise.  Although   Maurer's general setting \cite{MAU93}  allows for the distribution  $P(X,Y,Z)$
to be adversarially influenced, but  in the case of multiple instances of triplets $(X,Y,Z)$ received through noisy channels from a single randomness source,
they are not influenced by the adversary.

\vspace{1mm}
\noindent
{\em Constructions.}
We give two constructions of secure key agreement protocols,
 with strong and weak reliability.
 
The first construction is an efficient three round protocol that achieves the rate $1-\rho$ when $\rho_r+2 \rho_w < 1 - \rho$ when strong reliability is considered.
The second protocol is  an efficient one round protocol  that achieves  the rate $1-\rho_r$, but only for weak reliability.
Both protocols have constant size alphabet.
\remove{

Public discussion channels had  been considered in both wiretap  and SMT models.
In wiretap setting it was shown \cite{M92,AC93,MAU97} that a public discussion channel  substantially expands 
the range of scenarios in which secure communication is possible.  In particular  
secure communication becomes possible even if the eavesdropper channel is less noisy than the
main channel. Public discussion channel can be realized by using a message authentication protocol in the manual channel model \cite{NSS06}. %, or by invoking almost-everywhere broadcast protocol.
The communication complexity of PD will be kept low.

 {\bl
The key agreement protocol are considered over AWTP channel and AWTP-PD channel. For key agreement protocol with strong reliability, a three round protocol over $\Sigma$ that  asymptotically achieves the  bound over both AWTP and AWTP-PD channels. For key agreement protocol with weak reliability, there is one-round key agreement protocol over $\Sigma$ that achieves the bound over both AWTP and AWTP-PD channels. All the key agreement protocols are  
efficient (polynomial time) for the sender and receiver.
}

}

\subsection{Motivation and applications}

Considering adversaries at physical layer of communication  %his is a 
gives a realistic model of adversaries in wireless communication systems,
evidenced by growing  research in physical layer security  %has become a very active research area 
 \cite{ALCP09,BS13,MBL09,PTDC11} in recent years.  This research however 
 has been primarily in networking and engineering communities with emphasis  %of these works however, is 
 on 
tools and techniques, such as modulation techniques, multiple input, multiple output antennas \cite{TV05,MS11,BB11} and signal design, to  achieve  security 
goals that  is  ``informally" stated.

Our adversary model is motivated by physical layer adversaries in wireless communication, where 
entities interact  with its neighbour over a channel that can be ``partially controlled" by the adversary.
However our formulation of the problem can also be used to model networks that are partially controlled by an adversary.
The network between Alice and Bob can be  modelled by node disjoint paths between them.
The adversary
selects  two subsets of paths (possibly overlapping), some for eavesdropping  and some for tampering.
The goal  is to  for Alice and Bob to share a secret key.
A similar model of network is considered in SMT \cite{DDWY93,FW00,FFGV07,KS08} problem.
In its full generality and when $S_r \neq S_w$, the tampering is algebraic and by adding a noise vector on the set
$S_w$. However for $S_r = S_w$ (and more generally any component that is both read and written to),
 the tampering will be arbitrary as Eve can determine the noise component $z_i= x^2_i - x^1_i$, where $x^2_i $ and $x^1_i$ are the new and old (read) value of the component.

Using physical  layer properties of the system for secure communication has the interesting intuition that
massive surveillance would translate to the requirement of everywhere physical presence which would
significantly raise the bar for successful surveillance.
  \remove{??? Generating keys using physical layer location properties \cite{}
has been widely studied \cite{M92,MAU97,KR09}.....[WRITE]
??

If one can make Eve to have to have physical presence everywhere, 
the task of learning all communications will become much harder. This will vastly reduce Eve's massive surveillance capabilities.
% to ha systems individual, 
%the threat of mass surveillance and attack will be vastly reduced. 
%
Physical layer security has unique characteristics that can substantially contribute to the overall strengthening  of security.
One of the main components of cryptosystems is  %shared 
key generation. This process  in practice 
% Key generation in say SSL, 
use true random number generator of the underlying operating system such as Linux or Windows.
%In Universal Algorithm Substitution,  
Possibility of algorithm substitution at different stages of system development, and threat of "bugs" such as hardware Trojans,
suggests the need for randomness generators other than the ones offered by the underlying OS.
% massively  distributed algorithms such as linux RNG  cannot be trusted. 
One can always use an external TRNG, for example a Quantum RNG \cite{} or a Lava Lamp \cite{}. However these solutions are impractical in the ever increasing mobility and reliance on small devices.
}
This paper effectively shows that under the reasonable assumption that    adversary cannot fully  access the  physical layer 
communication system (channel or network), 
Alice and Bob can generate shared randomness, with provable security against a computationally unlimited adversary, and are able to securely communicate.
This provides an interesting new research direction for  using physical layer security as a source of individualization and diversification in security systems with the goal of improved security against massive surveillance.
\remove{
Physical layer security can be used to move the cryptosystems towards more individulaization and introduces a
rich source of diversity that 
%has the unique characteristics that moves the cryptosyetem 
can substantially contribute to the overall strengthening  of security,
by considering individuals' physical characteristics of the system.
}
A well-known approach to increasing   security against massive blanket  attacks in computer systems is
using diversification and individualization of software and hardware systems.  Diversification has been
successfully used by computer virus writers  to avoid detection by making  each copy distinct \cite{PETER05}. It has also been
used by security 
system designers, for example by using multiple operating systems and protection software,  to protect against 
mass infections.
In cryptography, massive surveillance through techniques such as algorithm substitution and backdoors, has motivated new research \cite{YY97,BPR14}.
% into  
%individualization and diversity of security technologies \cite{}.  
An important source of diversity in secure communication is physical layer properties of the communication systems.
%for providing security.  
To thwart systems that use physical layer properties of the system as a resource,
% effectively requires that 
the adversary must exert higher level of control over the  physically environment which  would be significantly more intrusive, visible and
demanding  on the adversary.
Our work is an step in this direction of exploring this resource for security against massive surveillance.

%%%%%%%%%%%%%%%%%%%%%%%%%%%%%
\subsection{Relation to previous work}
Maurer \cite{M92} considered a setup where Alice, Bob and Eve have correlated variables $X$, $Y$, and $Z$, distributed as $\bPr_{XYZ}$.
Alice and Bob want to share a shared secret key by exchanging messages over an authenticated insecure channel that is  fully readable by the adversary. 
He derived an upperbound on the entropy of the key that can be obtained from such a protocol,
\begin{equation} \label{eq1}
\bH(K) \leq  \min \{\bI(X;Y), \bI(X:Y|Z)\} +\bH(K,K')+ \bI(K;C^tZ)              
\end{equation}

To obtain more concrete results, he considered a scenario where a discrete memoryless channel generates sequences $X^n= (X_1\cdots X_n )$,   $Y^n=(Y_1\cdots Y_n )$, and $Z^n=(Z_1\cdots Z_n )$.
{\em  In such a setting,  rate of secret key agreement} is introduced \cite{M92}.  This is the maximum rate at which Alice and Bob can agree on a secret key, while the rate at which Eve obtains information is arbitrarily small. 
This definition was later \cite{MAU97} strengthened by requiring that the total leakage of the key approaches  zero, when the communication length $n$ over the channel approaches infinity.
They extend a  lower bound on the rate of secret key agreement in \cite{MAU97} to the key rate with strong secrecy. The bound states,
\begin{equation} \label{eq2}
R \geq \bI(X;Y) -\min \{ \bI(Z;X), \bI(Z;Y)\}  		
\end{equation}
In this model %In calculating rate, 
 communication over public discussion channel is assumed free.

\remove{
he channel provides authenticated communication: it is argued that message corruption can be detected if one assumes a short secret key between the sender and the receiver. The work thus shows that when Eve does not have a perfect knowledge of the communication, a short key can be extended to an arbitrarily long key.
}

\vspace{1mm}
\noindent
{\bf NOT Authenticated channels. }
In \cite{MAU97}  Maurer removed the assumption that the channel connecting Alice and Bob is authenticated.  He allowed the adversary to be able to control this channel and completely control the messages that are sent over it. The goal of the protocol is  to establish a shared key that is perfectly unknown to the adversary. This means that this goal could be achieved while some tampered communication remains undetected. 
Maurer considered two cases: (1). Alice, Bob and Eve do not share the initial information. (2). Alice, Bob and Eve share the initial information $X, Y, Z$ with distribution $\bPr_{XYZ}$.  Theorem 1 \cite{MAU97} shows that  in case (1) no secret key can be established. This intuitively says that if there is no initial common randomness and communications are completely tamperable, then no secret shared key can be expected. The theorem further shows  impossibility 
of key agreement even when all messages are authenticated (but public), or communication is authentic in one direction, and secret in the opposite direction. So without initial shared randomness, even if no tampering and only  one direction  visible by the adversary, one cannot expect a shared secret key. 
Thus initial setup is necessary for any information theoretic  secret key agreement protocol.

Maurer defines $X$-simulatable ($Y$-simulatable) and shows impossibility of key agreement if  $\bPr_{XYZ}$ is $X$-simulatable ($Y$-simulatable).
When %He considers the case that 
 the triplet $X^n$,  $Y^n$,  and $Z^n$  is generated by many applications of the same experiment, % and has ajoint product distribution,  
the secret key rate $S^*(\bPr_{XYZ})$ can be defined. 
A %He 
%proved the
surprising result of the paper is that  $S^*(\bPr_{XYZ})=0$, or $S^*(\bPr_{XYZ}) = S (\bPr_{XYZ})$, 
 and this distinction is based on or $ \bPr_{XYZ}$ is $X$-simulatable (or $Y$-simulatable ) or not.
 
That is if $\bPr_{XYZ}$ is not $X$-simulatable ($Y$-simulatable), 
using  non-authenticated communication gives the same rate  as using authenticated communication.

\vspace{1mm}

\noindent
{\bf $(\rho_r,\rho_w)$-Correlation. } 
To compare our results with the above models,  we consider a
distribution $\bPr_{X^n Y^n Z^n }$ that is generated by Alice sending $X^n= (X_1\cdots X_n)$ over a \rwc~ to Bob. Here $Y^n$ and $Z^n$ denote
the observations of Bob and Eve, under the restrictions imposed by the channel. We refer to such correlated variables, as $(\rho_r,\rho_w)$-triplet of length $n$.
Maurer notes  \cite{MAU97}  ``In  general the distribution $\bPr_{XYZ}$ may be under Eve partial control, and may only partly be known to Alice and Bob." %  Although he assumes $\bPr_{XYZ}$ is known by all parties.  
$(\rho_r,\rho_w)$-correlation generates $\bPr_{XYZ}$ under the influence of Eve.

We consider this initial  correlation in, (i) Maurer's setting of  \cite{M92} where communication is over a PD, and (ii)  Maurer setting of \cite{MAU97} where the channel is fully controlled by the adversary. In 
{\em $(\rho_r,\rho_w)$-correlation in setting of  \cite{M92}}, $\bPr_{X^nY^nZ^n}$ is  a $(\rho_r,\rho_w)$-correlation of length $n$, with interaction over PD and using strong reliability definition.
Maurer upper bound  for a key agreement protocol with security $\bI(K; Z) \leq \epsilon$ and reliability bounds $\delta$  is,
\[ \bH(S)  \leq  \min \{\bI(X;Y), \bI(X:Y|Z)\} + h(\epsilon)+\epsilon( |{\cal S}|-1) .\]

The key rate of a protocol  in Maurer's  setting however only considers the number $n$ of instances of the triplets $(X,Y,Z)$
shared by parties, and assumes free communication over public channel. That is the rate in fact is the expected amount of entropy that
can be derived from each instance.

The secret  key   in our setting however is the result of interaction over a {\em  partially} observed channel, and 
takes the total number of channel uses into account.

 The two bounds   give  the same result. 
 A  partially adversarially controlled channel can be used to generate secret keys at the same rate
  as  key agreement protocols in a setting of initially shared dependent vector variables, and having access  to a
  PD channel. 
  This means that the  requirement of the existence of PD can be 
  replaced with   channels that are partially controlled by the adversary, without incurring rate loss.
  
 % will relax the primitive of key agreement protocol, and allow us to design protocols for practical usage.

\noindent
{\em $(\rho_r,\rho_w)$-correlation in setting of  \cite{MAU97}.} 
Again $\bPr_{X^nY^nZ^n}$ is $(\rho_r,\rho_w)$-correlation, with interaction over a network layer adversary, and using weak definition of reliability. The question that one can ask is, if the 
distribution  $ \bPr_{X^nY^nZ^n}$ is $X$-simulatable ($Y$-simulatable),

\vspace{1mm}
\noindent
{\bf Adversary model.}
The adversarial model  for physical layer adversaries was first proposed in \cite{PS13}. The goal of 
the protocol however was to provide  reliable communication. Secure communication using the same adversary model
was studied  in \cite{WSN14,WSN141}. 
We use the same adversary model and refer to it as %\rwc~, or 
\rw- Adversarial Wiretap Channel (\rw-AWTP).
A secure message transmission protocol  
can be used to send a random key and establish a shared key.
The secrecy and reliability definition in \cite{WSN14,WSN141} ensure that the received random
string satisfies security properties of definition \ref{def_keyagree}  and so
can be used as a key.  It was proved that the secrecy rate $R_\sfM$ of a message transmission protocol (the highest  number of
bits per channel use, where transmission has perfect secrecy and reliability approaches perfect reliability with increased 
message length) is bounded by $R_\sfM\leq 1-\rho_r-\rho_w$.   This means that  no secure message
transmission is possible if $\rho_r=\rho_w=0.5$, that is $S_r$ and $S_w$ are 
chosen each to be half of the codeword, even if $S_r=S_w$ and half of the components of each round are left untouched.
The results of this paper shows that  secure communication is possible if $S_r=S_w$  and each set is  as large as $(1-\nu)n $, where $\nu$ is a negligible constant.

The adversary corrupts the codeword  by adding a noise vector,  with non-zero element over $S_w$, to the codeword.
If the adversary chooses $S_r=S_w$, then they know the component that is corrupted (because it is in $S_r$) and so
can design the noise to change the component to any desired value. That is arbitrarily change the 
component.
Our work shows possibility of secure communication  if  $(1-\nu)n $ components are arbitrarily corrupted.

\vspace{1mm}
\noindent
{\bf Other works.}
Key agreement is a fundamental problem with a very large body of research.
Directly related work on secure key agreement can be grouped into those  that assume 
a shared partly leaked string \cite{DKRS06,DORS08,KR09,DW09} at the start,  and those that do not assume a shared string, but assume
a close or highly correlated variable \cite{M92,MAU97,RW04,MW96,CFH15}.
In each group, communication can be over PD, or a tamperable channel.

Renner {\it et al.} consider the general setting of  key agreement protocol between Alice and Bob with variables,  $X$ and $Y$,
 that are similar but not identical, while Eve's information about $X$ and $Y$ is incomplete, with communication over  completely insecure channel.
 They  find bound and propose a protocol.
%  They also consider the construction of key agreement protocol over an interactive and insecure channel. 
Kanukurthi {\it et al.} \cite{KR09} propose an efficient  key agreement protocol in the same setting.

This problem has also been considered under robust fuzzy extractors \cite{DORS08} and \cite{DKRS06}  for also the case that
parties may have a long-term small secret key.

Extracting secret and shared randomness (key) when parties have weak, partially leaked secret
has been studied in \cite{BBCM95}. 
Maurer {\it et al.} consider  privacy amplification against passive and active adversaries with
 incomplete information about a shared string between two parties. 

Dodis {\it et al.}\cite{DW09}  give a two round protocol that  optimally extracts the randomness in a 
shared string of length $n$, by interacting over a channel that is controlled by the adversary.
\remove{
consider the key agreement protocol that Alice and Bob share
an n-bit secret $x$, which might not be uniformly random,
but the adversary has at least $k$ bits of uncertainty about it. They propose an authenticated key agreement protocol in which Alice and Bob use $x$ to agree on a
nearly uniform key $R$, by communicating over a public channel controlled by an active adversary Eve. 

Yevgeniy Dodis and Daniel Wichs. Non-malleable extractors and symmetric key cryptography
from weak secrets. 
}

 The work in \cite{CFH15}  gives characterization of distributions $\bPr_{XYZ}$  in
different  communication setting such as when there is no communication, there is a one-way communication and there
is a helper.

 \section{Definitions of Key Agreement Protocol}

\subsection{Channel Models}

Let $n$ be the length of codeword, $[n]=\{1,\cdots, n\}$. We denote set $S_r= \{i_1,\cdots, i_{\rho_r}
n\} \subseteq [n]$ and $S_w= \{j_1,\cdots, j_{\rho_w}n\} 
\subseteq [n]$ be the two subsets of the $n$ coordinates. Let $\mathsf{SUPP}(x)$ of vector $x\in \Sigma^n$ be the 
set of positions in which the component $x_i$ is non-zero.

\begin{definition}\label{def_1awtp}  A (($\rho_r, \rho_w$)-Adversarial Wiretap 
Channel(($\rho_r, \rho_w$)-AWTP))
is an adversarially corrupted communication channel 
between Alice and Bob such that it is (partially) controlled 
by an adversary Eve, with two capabilities: Reading and 
Writing. For a codeword of length $n$, Eve can do the 
following:
\begin{enumerate}
\item Reading (eavesdropping): Adversary selects a subset $S^r \subseteq 
[n]$ of size $|S^r| \leq \rho_rn$ and reads the components of 
the sent codeword $c$ on $S^r$. 
\item Writing (modifying): Adversary chooses a subset $S^w \subseteq 
[n]$ of size $|S^w| \leq \rho_wn$ for writing, and adds to $c$ 
an error vector $e$ with $\mathsf{SUPP}(e) = S^w$. 
\end{enumerate}
\end{definition}

For each channel,  the subset $S=S^r\cup S^w$ with size $|S| \leq
\rho n$ is  the set of codeword components   that 
are  either read   or write to, by the adversary.
%In each round of AWTP transmission, 
%We assume that the reading parameter $\rho_r$, writing parameter $\rho_w$, and reading or writing parameter $\rho$, are fixed. 
We assume the adversary is adaptive and   selects
components of the codeword  for reading and writing one by one, using its current view of the communication,
% and also 
%from his previous corruption, at each step using 
%his knowledge of the codeword at that time.
 In each communication round, 
the two   subsets $S^r$ and $S^w$, chosen by Eve,  may be different but will satisfy the bounds    $|S^r|\leq \rho_rn$,    $|S^w|\leq \rho_wn$ and $|S|\leq 
\rho n$.

A key agreement protocol is an interactive protocol that uses 
% we allow  the 
%communication over two  
the $(\rho_r,\rho_w)$-AWTP channel  is two 
directions, 
%The transmission over AWTP channel can either 
from Alice to Bob, or form Bob to Alice. 

%We consider two types of channels that connecting between Alice and Bob: One channel is called {\it public discussion channel}. The adversary can read but cannot tamper the communication on public discussion channel. 
%The second channel is called {\it adversarial wiretap channel}. The adversary can partially read and tamper the communication on the adversarial wiretap channel.

To compare our results with previous ones, we also consider  {\it public discussion channels}. 
\begin{definition}\label{def_pd} A (Public Discussion Channel(PD)) is an %insecure and 
authenticated channel between Alice and Bob, that can be 
read by everyone including Eve. The adversary's reading capability is $S_r = [n]$, while the writing capability is $S_w = 0$. 
\end{definition}

\subsection{Key Agreement Protocols}

%We consider two key agreement protocols: The first key agreement protocol is over one-way AWTP channel and interactive PD channel. The second key agreement protocol is over interactive AWTP channel. 

We study key agreement protocols over \rw-AWTP channels. We consider the case that the channel is the only resource, and 
interaction over this channel generates the key.

%\vspace{3mm}
%{\bf  Interactive Key Agreement Protocol overAWTP Channel.}

\subsubsection{ Interactive Key Agreement Protocols over \rw-AWTP Channels.} \label{interactive}
There are a pair of   forward    and  backward  \rw-AWTP channels, from Alice to Bob and  Bob to Alice, respectively (Figure \ref{fig_keyagreeawtp}). To establish a secret key, Alice and Bob follow the %key agreement 
an $\ell$-round key agreement protocol, sending coded messages over the two channels. %in  communication rounds. 
%In each round of  key agreement protocol, Alice sends information to Bob, or Bob sends information to Alice over AWTP channel. 
%
  The protocol is defined by a sequence of randomized   function pairs $ (\Pi_A^r,\Pi_B^r)$ for $r = 1, \cdots, \ell$,  
  and a pair of deterministic key derivation functions $(\Phi_A, \Phi_B)$. 
   Each protocol function outputs a vector over alphabet symbols $\Sigma$.
  In the $i^{th}$ round, Alice transmits the protocol message  $c_i$ to Bob, or Bob transmits the protocol message $d_i$ to Alice. 
 %In each round of communication, 
 Eve reads and writes to the  channel.
In the $i^{th}$ round, Eve reads on the set $S^r_{i}$ and adds error on the set $S^w_{i}$, and the sizes of $S^r_{i}$ and $S^w_{i}$ are  bounded  as
 $|S^r_{i}| \leq \rho_rn_{i}$, $|S^w_{i}|\leq \rho_wn_{i}$, respectively. 
At the end of the $i^{th}$ round, Bob (or Alice) receives a corrupted word $x_{i}$  (or $y_{i}$).

  Let $r_A$ and $r_B$ denote the randomness of Alice and Bob, and $ v_A^{i}$ and $v_B^{i}$ denote  the  views of Alice and Bob, 
  respectively. The view of Alice $v_A^i$ consists of all messages  received and sent by her at the end of round $i-1$,

\begin{equation}
c_{i}=\Pi_A(r_A, v_A^{i})   \mbox{\qquad and \qquad } d_{i}=\Pi_B(r_B, v_B^{i}).
\end{equation}

\begin{figure}[h]
\caption{Key Agreement Protocol over Interactive AWTP Channel}
\centering
%\vspace{3mm}
\includegraphics[width=0.5\textwidth]{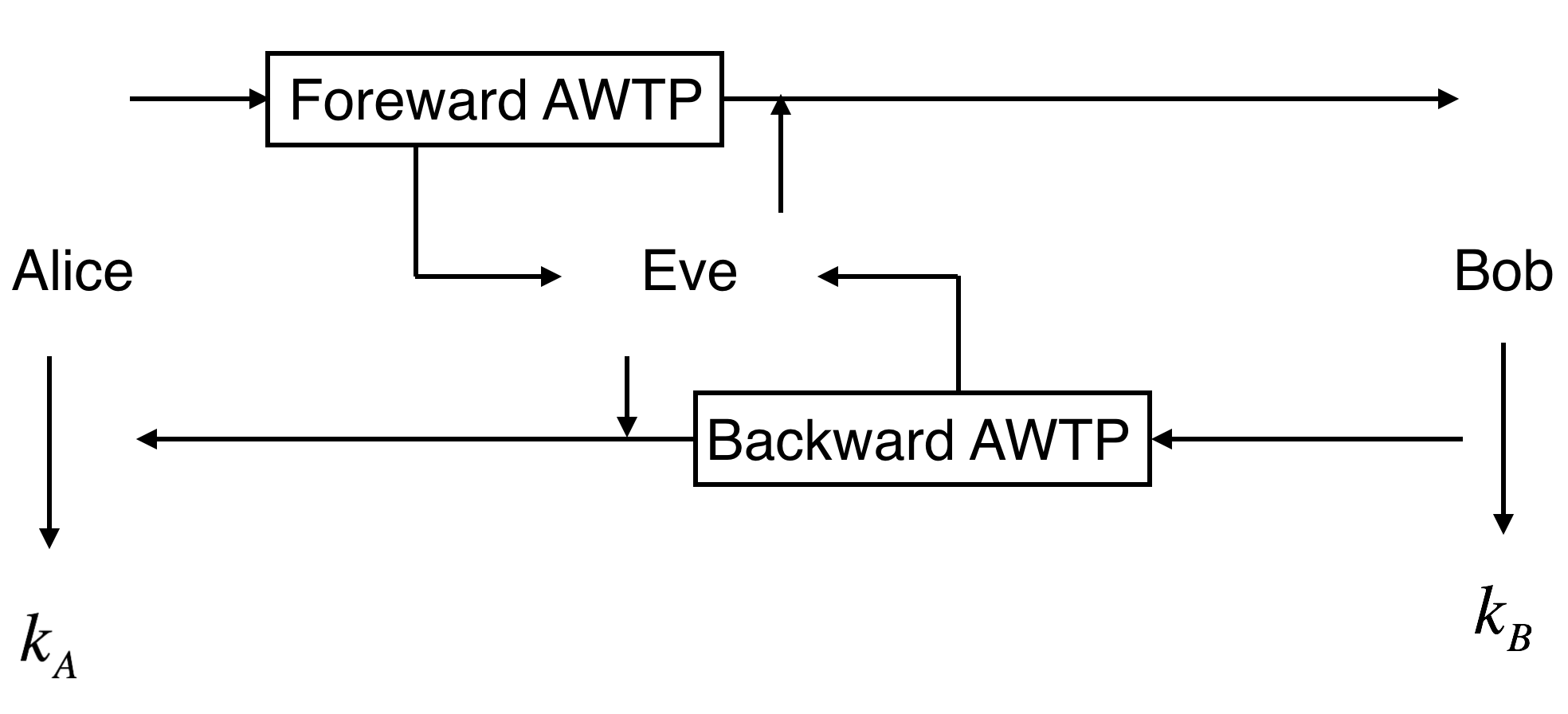}
\label{fig_keyagreeawtp}
\end{figure}

%%%%%%%%%%%%%%%%%%%%%%%%%%%

At the end of the $\ell^{th}$ round, Alice and Bob  generate the   keys $k_A$ and $k_B$,
using their sent and received  messages that form  their views of the protocol,
% based on the communications of protocol. 
\begin{equation}\label{rel-str}
\Phi_A(v_A^{\ell})=k_A \mbox{\qquad   and  \qquad} \Phi_B(v_B^\ell)=k_B.
\end{equation}
The key derivation algorithm is deterministic.
Since there is no initial dependent variables held by %, and  between the randomness of 
Alice and Bob,  the key will only depend on the communication transcripts. %and will not contribute to the key generation process. 

We also consider a second case that the key derivation function outputs, either a key or  detects an error and outputs $\perp$.
That is,
\begin{equation} \label{rel-perp}
\Phi_A(v_A^{\ell})=k_A \mbox{\qquad   and  \qquad} \Phi_B(v_B^\ell)=k_B  
\end{equation}

%%%%%%%%%%%%%%%%%%security and reliability%%%%%%%%%%%%%%%%%%

{\bf Security and Efficiency.}  
The protocol has $\ell$ rounds.
%The total number of communication  rounds of the 
%key agreement protocol is $\ell$.  
%In  each round, the length of the codeword transmitted by Alice or Bob over the \rw-AWTP 
%or PD channel
 %s $n_{i}$ protocol symbols,  where $i = 1, \cdots, \ell$.
In the $i^{th}$ round of communication, %key agreement protocol, 
 the length of the protocol message  %over AWTP channel
  is $n_{i}$  protocol alphabet,  where $i = 1, \cdots, \ell$. 
 The   total length of communication is $n=\sum_{i=1}^\ell n_{i}$.

The following gives  correctness, security and reliability definitions  of the key agreement protocol. 

\begin{definition}\label{def_keyagree}
{\em $(\epsilon, \delta)$-Secure Key 
Agreement  $($$(\epsilon, \delta)$-$\mathsf{SKA}$$)$ Protocol: }\label{def_keyagreepub}
%Alice, Bob and Eve that are connected by a \rw-AWTP forward and backward 
%channels. The 
An $(\epsilon, \delta)$-key agreement protocol  satisfies the following properties:

\begin{enumerate}

\item Correctness: If Eve is passive, that is $S_w = \emptyset$, then $\bPr(K_A = K_B) = 1$.

\item Secrecy: Let $U$ be  a uniformly distributed variable over $
\cal K$. For any adversary  view $Z$, the statistical 
distance between the distribution of key and $U$ is 
bounded by $\epsilon$. That is,

\begin{equation}\label{eq_sec}
\SD(P_{K|Z}, U)\leq \epsilon \qquad
\end{equation}

\item Strong Reliability: If the protocol key derivation function follows (\ref{rel-str}),  The probability that Alice or Bob output  different keys  is bounded by $\delta$. That is,
 Alice and Bob output 
a common key  with  probability  at least $1-\delta$,  
\begin{equation}\label{eq_rel}
\bPr(K=K_A=K_B)\geq 1-\delta
\end{equation}
\item Weak Reliability: 
If the protocol key derivation function follows (\ref{rel-perp}),  Eve wins is  if $K_A \neq \perp, K_B \neq \perp, K_A \neq K_B$. That is, the probability that Alice and Bob either output error, or output the correct key, is at least $1 - \delta$,
\begin{equation}\label{eq_wrel}
\bPr(K_A = \perp, \mbox{or } K_B = \perp, \mbox{or } K = K_A = K_B )\geq 1-\delta
\end{equation}

\end{enumerate}
\end{definition}

\vspace{3mm}

The key agreement protocol is {\em perfectly secure} if $
\epsilon=0$, and {\em perfectly reliable} if $\delta=0$. 

The transmission efficiency of a key agreement protocol   is 
measured by the  {\em secret rate $R_{\sk}$ } which is the rate at which a secret key is agreed between Alice and Bob,
assuming  the adversary  uses their best possible adversarial strategies.
 For a protocol with total transcript length $n$ the secret key rate is given by, $\frac{\log|\cal K|}{n\log |\Sigma|}$.

\remove{
We also consider a {\em protocol family} $\{ \{\Pi^u_A, \Pi^u_R\}, u \in \mathbb{N}\} $ where $\{\Pi^u_A, \Pi^u_R\}$ is a protocol for which
$|{\cal K}| =u$.
}

The rate of key agreement 
protocol is the maximum rate of key that Alice and 
Bob can generate by communicating over AWTP channel. A key agreement  is parameterd with the total length $n$ of protocol. 

\vspace{3mm}
\begin{definition}\label{def_keyagreefamily}
The rate $R_\sk$ of key agreement protocol is achievable for protocol $\mathsf{SK}=\{\Pi_A^r, \Pi_R^r\}$ for $r = 1, \cdots, \ell$, if for any $\xi>0$, there exist $n_0$ such 
that for any $n\geq n_0$, there is 
\[
\frac{\log|\cal K|}{n\log |\Sigma|}\geq R_\sk-\xi
\]
and,
\[
\delta<\xi
\] 
\end{definition}
\vspace{3mm}

 The $\epsilon$-secret key capacity $\mathsf{C}_\sk^\epsilon$ of key agreement protocol 
 is the largest achievable rate of all key agreement 
protocol over the channel with $\epsilon$-secrecy. The 
perfect secret key capacity ${\sfC}_\sk^0$ is the 
largest achievable secret key rate of all key agreement protocols over 
the channel. 
% with perfect secrecy.

The computational efficiency of a key agreement protocol 
refers to the computational complexity of  the protocol algorithms run by Alice and 
Bob. The key agreement protocol is efficient if both parties 
computations are  polynomial time. Otherwise, the key 
agreement protocol is inefficient.

%{\bf Key Agreement Protocol over AWTP-PD Channel}
\subsubsection{Key Agreement Protocol over AWTP-PD Channels}

 To better compare and contrast our results with the known results in key agreement, and in particular 
 Maurer's setting in \cite{M92,MAU93},  we 
will consider a setting where Alice and Bob have access  to a one-way \rw-AWTP channel,
 and a  two-way PD channel.  They will use the \rw-AWTP channel to establish dependent and partially 
 leaked variables, and then use the PD channel to extract the entropy captured in the established dependent variables.

There is a one-way AWTP channel from Alice to Bob, and two-way PD channels between Alice and Bob (Figure \ref{fig_keyagreeawtppd}). To establish a secret key, Alice and Bob follow an $\ell$ round  key agreement protocol. In the first round of  the key agreement protocol, Alice sends a sequence of variables $X^n$ to Bob over  the \rw-AWTP channel.
 In the following rounds, Alice and Bob communicate %with each other 
 over the PD channel.

\begin{figure}[h]
\caption{Key Agreement Protocol over AWTP-PD Channel}
\centering
\vspace{3mm}
\includegraphics[width=0.5\textwidth]{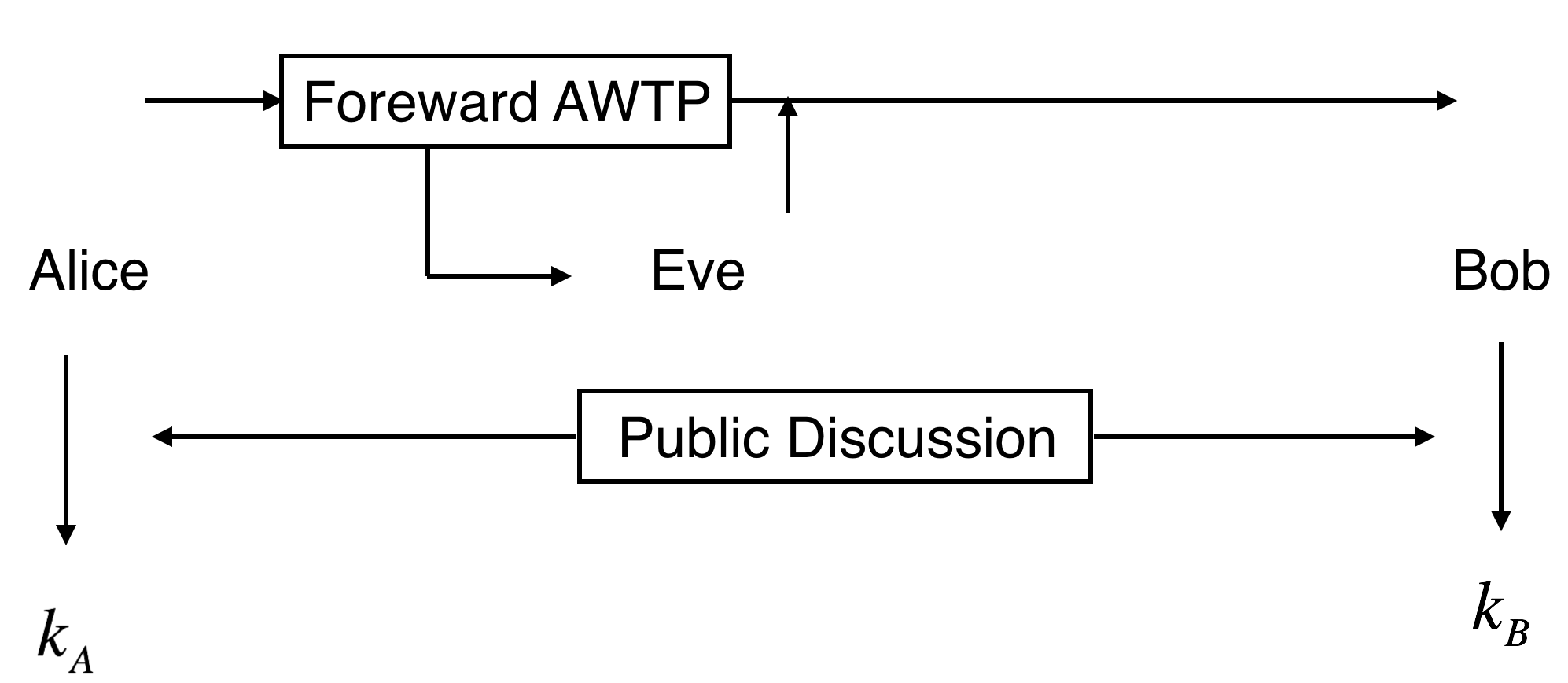}
\label{fig_keyagreeawtppd}
\end{figure}

Description of the protocol messages and key derivation functions, and the  definitions of correctness, security and reliability are the same
as the interactive case above. The only difference  is that the \rw-AWTP channel is used only in the first round.
\remove{
As before,
the protocol is defined by a sequence of randomized function pairs $(\Pi_A,\Pi_B)$, and a pair of deterministic key derivation functions $(\Phi_A, \Phi_B)$. In the $i^{th}$ round, Alice transmits the  protocol message  $c_i$ to Bob, or Bob transmits the   protocol message $d_i$ to Alice. Let $r_A$ and $r_B$ be the randomness of Alice and Bob, and $v_A^{i}$ and $v_B^{i}$ be Alice and Bob   views of at the end of communication 
round $i-1$.

\begin{equation}
c_{i}=\Pi_A(r_A, v_A^{i})   \mbox{\qquad and \qquad } d_{i}=\Pi_B(r_B, v_B^{i})
\end{equation}
}

%Eve's read and write capabilities over AWTP channel in the first round of  the key agreement protocol, is the same as defined for 
%interactive case.
In the first round of  the protocol, Eve reads on the set $S_{r}$, and adds error on the set $S_{w}$ of  the \rw-AWTP channel %, with 
and we have %the size of $S_{r}$ and $S_{w}$ bounded by 
$|S_{r}| \leq \rho_rn_1$, $|S_{w}|\leq \rho_wn_{1}$. 
At  the end of the first round, Bob receives a corrupted word $x_{1}$,  and Eve has the selected partial view given
by  $z_1$. % over AWTP channel. 
In the following rounds, communication is
over PD and is fully accessible to Eve.
% can only read over PD channel in the following round.
 This means that in the $i^{th}$ round of  the %key agreement 
 protocol with $i \geq 2$, Eve read and write sets are   $S_r = [n_i]$ and %writing on set 
 $S_w = \emptyset$, with $|S_r| \leq n_i$ and $|S_w| = 0$.
That is at the end of $i^{th}$ round, $i\geq 2$, Bob or Alice correctly  receives the sent codeword. % as Alice or Bob sends over PD channel. 

%%%%%%%%%%rate of key agreement%%%%%%%%%%%

\section{Rate Bounds} % Key Agreement Protocol}
\label{SecRate}

 We derive upper bounds on the secret key rate of  key agreement protocols over \rw-AWTP channel.

\subsection{Interactive key agreement}
 This is the main setting considered in this paper. The only resource available to the adversary
is the channel that is partially controlled by the adversary.
We consider % first give the rate of key agreement protocol with
 strong reliability. 

\begin{theorem}\label{the_boundkey}
The upper bound on the secret key rate of an  interactive key agreement protocol  over a \rw-AWTP channel, 
 is $R_{\sk}\leq 1-\rho$.  The bound is for  strong reliability.
\end{theorem}

\begin{proof}
We assume the %total round of commutation of key agreement 
protocol has $\ell$ rounds.  
The length of the protocol message in the $i^{th}$ invocation of the AWTP channel  is $n_i$, and  $[n]=\bigcup_{i=1}^{\ell} [n_i]$. 
 For the $i^{th}$  communication round, let $c_i$ and $d_i$ be the codewords sent by Alice and Bob, respectively; $c_{i,j}$ and $d_{i,j}$ denote the $j^{th}$ components of codeword $c_i$ and $d_i$, respectively; $c^i$ and $d^i$ denote concatenations of all codewords sent in all invocations up to, and including, the $i^{th}$ round transmission. We use capital letters to refer to the random variables associated with,
$c_i$,  $d_i$, $c_{i,j}$,  $d_{i,j}, c^i$ and $d^i$, as $C_i, D_i, C_{i, j}, D_{i, j}, C^{i}$ and $D^{i}$, respectively.  
Let $C^{\ell, r}$, $C^{\ell, w}$, $D^{\ell, r}$, $D^{\ell, w}$ be the random variables of the  protocol messages on the adversarial reading sets and writing sets. 
Let 
$C^{\ell, a}$ and $D^{\ell, a}$ be the random variables on adversarial read only sets, $C^{\ell, b}$ and $D^{\ell, b}$ be the random variables on the adversarial read and write sets, $C^{\ell, c}$ and $D^{\ell, c}$ be the random variables on adversarial write only sets, $C^{{\ell}, d}$ and $D^{\ell, d}$ be the random variables corresponding to the sets without adversarial corruptions, respectively. We use  $X_i$ and $Y_i$ to  denote the corrupted word received by Bob and Alice when $C_{i}$ and  $D_{i}$ are sent
(by Bob and Alice, respectively), and define similarly $X_{i,j}, X^i$  %$X_{i}, X_{i,j}, X^i$
 and  $Y_{i,j}, Y^i$ %$Y_{i}, Y_{i,j}, Y^i$ 
 corresponding to $C_{i,j}, C^i$ and $D_{i,j}, D^i$, respectively.  

{\bf Step 1:} 
We first define an adversary $Adv_1$ that works as follow:
\begin{enumerate}
\item Selects the reading sets and writing sets in all $\ell$ rounds before the start of the protocol.
\item During the protocol, in  round $i$, %the adversary
 chooses a  random error vector $e$, and adds it    to the set $S^w_i$ of that round. % adversarial writing sets.
\item During  round $i$, the adversary reads the   components of  $S^r_i$.   %adversarial reading sets.
\end{enumerate}
Note that this adversary does not  use the information seen during the protocol to improve their chance of 
making the protocol to fail.
We give two lemmas that follow from secrecy and reliability. 

\begin{lemma}\label{le_sec}
For an $(\epsilon, \delta)$-key agreement protocol, the following holds for adversary $Adv_1$:
\begin{equation} \label{pf_bound1}
\sfI(K;C^{\ell,r}D^{\ell,r})\leq 4\epsilon  n\log(\frac{|\Sigma|}{\epsilon})
\end{equation}
and,
\bequation \label{pf_bound2}
\log |\cK| - \sfH(K) \leq 4\epsilon n \log(\frac{|\Sigma|}{\epsilon})
\eequation
\end{lemma}

Proof is in Appendix \ref{ap_le_sec}.

\begin{lemma}\label{le_rel}
For  any $(\epsilon,\delta)$-key agreement protocol,  the following holds:
\begin{equation}
\sfH(K|C^{\ell,a}C^{\ell,d}D^{\ell,a}D^{\ell,d})\leq 2\sfH(\delta)+2\delta\log |\cK|
\end{equation}
\end{lemma}

Proof is in Appendix \ref{ap_le_rel}.

{\bf Step 2:} We prove the upper bound: 
\[
\frac{\log |\cK|}{2n\log |\Sigma|}\leq 1-\rho+2\epsilon (1+\log_{|\Sigma}\frac{1}{\epsilon})
\]
We have,
\begin{equation}
\begin{split}
\sfH(K)&=\sfI(K; C^{\ell,a}C^{\ell,b}D^{\ell,a}D^{\ell,b}) + \sfH(K | C^{\ell,a}C^{\ell,b}D^{\ell,a}D^{\ell,b}).
\end{split}
\end{equation} 

The  first item $\sfI(K; C^{\ell,a}C^{\ell,b}D^{\ell,a}D^{\ell,b})$ is upper bounded using  Lemma \ref{le_sec} (Eq. (\ref{pf_bound1})). For second item,  %\\
  $\sfH(K | C^{\ell,a}C^{\ell,b}D^{\ell,a}D^{\ell,b})$, we have,

\begin{equation}
\begin{split}
&\sfH(K | C^{\ell,a}C^{\ell,b}D^{\ell,a}D^{\ell,b})\\
&=\sfH(K C^{\ell,b}D^{\ell,b} | C^{\ell,a}D^{\ell,a})- \sfH(C^{\ell,b}D^{\ell,b} | C^{\ell,a}D^{\ell,a})\\
&=\sfH(K | C^{\ell,a}D^{\ell,a}) + \sfH( C^{\ell,b}D^{\ell,b} |  K C^{\ell,a}D^{\ell,a}) - \sfH(C^{\ell,b}D^{\ell,b} | C^{\ell,a}D^{\ell,a})\\
&=\sfH(K | C^{\ell,a}D^{\ell,a}C^{\ell,d}D^{\ell,d}) + \sfH(C^{\ell,d}D^{\ell,d} | C^{\ell,a}D^{\ell,a}) - \sfH(C^{\ell,d}D^{\ell,d} | KC^{\ell,a}D^{\ell,a}) \\
&\;\;\;\;+ \sfH( C^{\ell,b}D^{\ell,b} |  K C^{\ell,a}D^{\ell,a}) - \sfH(C^{\ell,b}D^{\ell,b} | C^{\ell,a}D^{\ell,a})\\
&\leq 2\sfH(\delta)+2\delta\log |\cK| +\sfH(C^{\ell,d}D^{\ell,d})\\
&\leq 2\sfH(\delta)+2\delta\log |\cK| + (1-\rho) 2n \log |\Sigma|.
\end{split}
\end{equation} 

So the bound on $\sfH(K)$ is,
\bequation \label{pf_bound3}
\sfH(K)\leq  4\epsilon  n\log(\frac{|\Sigma|}{\epsilon}) + 2\sfH(\delta)+2\delta\log |\cK| + (1-\rho) 2n \log |\Sigma|.
\eequation
From (Eq. (\ref{pf_bound2}) and (\ref{pf_bound3})), and letting $\delta \rightarrow 0$ as $n\rightarrow \infty$, we have,
\bequation
R_\sk=\frac{\log |\cK|}{2n \log |\Sigma|}\leq 1-\rho + 4\epsilon (1+\log_{|\Sigma|}\frac{1}{\epsilon}).
\eequation
\qed
\end{proof}

 \vspace{2mm}
 \noindent
 {\bf Weak reliability.}
%Next, we show the upper bound on rate of key agreement protocol with weak reliability over AWTP channel.
A natural question is if the upper bound will be  affected if strong reliability is replaced by weak reliability, where 
the protocol success also includes cases that Alice and/or Bob output $\perp$.
We prove the following theorem using an approach similar to above. The proof is in full version of paper. 
\begin{theorem}\label{TheBoundKeyPD}
The upper bound on rate of key agreement protocol with weak reliability over AWTP channel is bounded by $R_{\sk} \leq 1 - \rho_r$.
\end{theorem}

The bound suggests that the corrupted components of   protocol messages (corresponding to $S_i^w$)  can be detected and 
so the secret key rate is  limited by the leakage of the components in $S_i^r$.

\subsection{Rate of Key Agreement Protocol over AWTP-PD channel}
In this setting the \rw-AWTP channel is used to establish the initial dependent variables, and the remaining construction is
over a two-way PD.
The following theorem gives the %First we give the 
upper bound on the secret key rate of key agreement protocol in this setting.
% with strong reliability over AWTP-PD channel. 
Proof strategy is similar to above and is given in the full version of the paper.
\begin{theorem}\label{TheBoundKeyPD}
The upper bound on the secret key rate of key agreement protocol over \rw-AWTP-PD channel, with strong reliability is,
% is bounded by
 $R_{\sk} \leq 1 - \rho$.\\
 For weak reliability the upper bound is $R_{\sk} \leq 1 - \rho_r$.
\end{theorem}

\remove{
The proof is in full version of paper. 

\vspace{3mm}

Next, we show the upper bound on rate of key agreement protocol with weak reliability over AWTP-PD channel.

\begin{theorem}\label{TheBoundKeyPD}
The upper bound on rate of key agreement protocol with weak reliability over AWTP channel is bounded by $R_{\sk} \leq 1 - \rho_r$.
\end{theorem}

The proof is in full version of paper. 

}

\begin{remark}
The rate in \rw-AWTP-PD does take into account the communication over the PD. This is different from Maurer's definition \cite{Mau93?} where PD is free. The bound however is the same as the interactive case  where all communications is over \rw-AWTP. 
This is  surprising and shows that  the secret key rate could could stay the same if the channel is partially corrupted.
\end{remark}

%%%%%%%%%key agreement protocol%%%%%%%%%%%

\section{Constructions} %and Secure Message Transmission Protocol}

We first introduce the building blocks that are use in our construction, %key agreement protocol, 
and then describe the constructions  of protocols that achieve the upper bounds   in Section 3.1  and 3.2. 
%and the construction of secure message transmission protocol. 

\subsection{Cryptographic primitives}

\subsubsection{Universal Hash Family}\label{sec_mac}

%Universal hashing family refers to select a hash function randomly from a family of hash functions with a certain mathematical property. This guarantees a low number of collisions in expectation, even if the message is chosen by adversary.

%\begin{definition}\cite{S02}
An $(N,n,m)$-hash family is a set  $\cal F$ of $N$ functions, $f: {\cal X}\rightarrow {\cal T}$,  $f\in \cal F$, where $|{\cal X}|=n$ and $|{\cal T}|=m$. 
%\end{definition}
%There will be no 
Without loss of generality, we assume $n\geq m$.
\begin{definition}\cite{S02} 
Suppose % that the $(N, n, m)$-hash family $\cal F$ has 
the range $\cal T$  of an $(N, n, m)$-hash family $\cal F$ %which 
is an additive Abelian group. 
$\cal F$ is called $\epsilon$-$\Delta$ universal, if  for any two elements $x_1, x_2\in {\cal X}, x_1\neq x_2,$, and for any element $t\in {\cal T}$, there are  at most $\epsilon N$ functions $f\in \cal F$ such that $f(x_1)-f(x_2)=t$, were the operation is from the group.
\end{definition}

Let $q$ be a prime and $u\leq q-1$. Let the message be ${\bf x}=\{x_1, \cdots, x_{u}\}$. For $\alpha \in \bF_q$, define the universal hash function $\mathsf{hash}_\alpha$ by the rule,
\begin{equation}\label{eq_mac1}
t=\mathsf{hash}_\alpha({\bf x})=x_1\alpha+x_2\alpha^2+\cdots+x_{u}\alpha^{u} \mod q 
\end{equation}
Then $\{\mathsf{hash}_\alpha(\cdot): \alpha\in \bF_q\}$ is a $\frac{u}{q}$-$\Delta$ universal $(q, q^u, q)$-hash family. This %We will use a classic 
is a known construction of $\frac{u}{q}$-universal hash family \cite{S02}.

\subsubsection{Message Authentication Code}%\label{sec_mac}

A message authentication code (MAC)\cite{S84}  is a  cryptographic primitive that allows a sender who shares a secret key with 
the receiver to construct authenticated messages  to be sent %send an information block 
 over a channel that is tampered by an adversary, and 
the receiver 
to  be able to verify the integrity of the received message.  

\begin{definition}
A message authentication code  consists of
two algorithms $(\mathsf{MAC},\mathsf{Ver})$ that are used for authetication and verification, respectively. 
For a message $m$  an {\em  authentication tag,}  or simply a {\em tag}, is  computed,
\[
t = \mathsf{MAC}(m, r),
\]
and a  tagged message $(m,t)$ is constructed. % to be sent over the channel.  
The verifier accepts a tagged  pair $(m,t)$ if $\mathsf{Ver}((m, t), r)) = 1$.
Security of  a one-time MAC is  defined as,
$$\Pr [(m', t'), \mathsf{Ver}((m',t'), r) =1|(m, t),  t= \mathsf{MAC}(m, r) ]\leq \delta$$
\end{definition}

We use a MAC construction that uses polynomials over $\bF_{q}$. Let  $m$ be a vector of length $\ell$, and $r = (\alpha, \beta)$, $t$ is over $\bF_{q}$. Define the MAC generation function $\MAC : \bF_q^{\ell} \times \bF_q^2 \rightarrow \bF_q$, 
where  $t = \MAC(m, (\alpha, \beta)) $ as,
\[
t=\MAC(m, (\alpha, \beta))= \sum_{i=0}^{\ell-1} x_i \alpha^{i} + \beta \mod q.
\]

\begin{lemma}\label{le_mac}
For the MAC construction above, the success probability of the adversary in forging a tagged message  $(m', t')$ that pass MAC verification is no more than $\frac{\ell}{q}$.
\end{lemma}
The proof is a direct extension of the proof in \cite{DGMP92}.

\subsubsection{Algebraic Manipulation Detection Code}

Algebraic manipulation detection code (AMD code) \cite{CDFP08} can be used to encode a source into a value  stored on $\Sigma(G)$ so that any tampering by an adversary will be detected, except with a small error probability $\delta$.

\begin{definition}[AMD Code \cite{CDFP08} ]\label{def_amd} 
Let   $\cal G$ be an additive group. An $({\cal X}, {\cal G}, \delta)$-Algebraic Manipulation Detection code ($({\cal X}, {\cal G}, \delta)$-AMD code) consists of two algorithms $(\mathsf{AMDenc}$ and $\mathsf{AMDdec})$ that are used for encoding and decoding, respectively. % of messages from $\cal X$.
Encoding is a probabilistic mapping $\mathsf{AMDenc}: {\cal X} \rightarrow {\cal G}$ that   maps an element of $\cal X$ to an element of the group $\cal G$. 
Decoding is a deterministic mapping $\mathsf{AMDdec}: \cal G \rightarrow {\cal X} \cup \{\perp\}$ and for any $x \in {\cal X}$ satisfies $\mathsf{AMDdec}(\mathsf{AMDenc}(x)) = x$.  The security of AMD codes requires,
\begin{eqnarray} \label{amd}
\bPr[\mathsf{AMDdec}(\mathsf{AMDenc}(x) + \Delta) \not\in \{x,\perp \}] \leq \delta, \,\, 
\end{eqnarray}
for all $x\in {\cal X}, \Delta \in {\cal G}.$
\end{definition}

An  AMD code is {\em systematic} 
if  the encoding has the form $\mathsf{AMDenc}: {\cal X}\rightarrow {\cal X} \times {\cal G}_1 \times {\cal G}_2$, mapping $x$ to $ (x, r, t=f (x, r))$
for some function $f$, where $r \stackrel{\$}\leftarrow {\cal G}_1$. % here $r\in {\cal G}_1$ is  the  randomness used in the encoding and $t\in {\cal G}_2$ is called  the {\em tag}. % of AMD code. 
The decoding function outputs $\mathsf{AMDdec}(x, r, t) = x$ if and only if $t=f (x, r)$, and $\perp$ otherwise.

We use a systematic AMD code, given in \cite{CDFP08}, over an extension field.  Let $d$ be an integer such that $d + 2$ is not divisible by $q$. Define the encoding $\mathsf{AMDenc}:\bF_{q}^d \rightarrow \bF_{q}^d \times \bF_{q} \times \bF_{q}$ as $\mathsf{AMDenc}(x) = (x,r,f(x,r))$, where:
\begin{equation}\label{EqAMDFunction}
f(x, r)= \left( r^{d+2}+\sum_{i=1}^{d} x_i r^i\right)\mod q.
\end{equation}

\begin{lemma}\label{le_amd}
For the AMD code above, the success chance of an adversary in tampering with a stored codeword
%that  has no access to the  codeword
 $(x,r,t)$ and  constructing a new codeword  $(x',r',t')= (x'=x+\Delta x, r'=r+\Delta r, t' =t+\Delta t)$,
that  satisfies $t' = f(x', r')$, is no more than $\frac{d+1}{q}$. 
\end{lemma}

\subsubsection{Randomness Extractor}\label{sec_extractor}

A randomness extractor is a function, which is applied to a weakly random entropy source (i.e.,  a non-uniform random variable), to
%a highly random output that appears independent from the source and
obtain a uniformly distributed source.

\begin{definition}\cite{DORS08}
A (seeded) $(n, m, r, \delta)$-strong extractor is a function 
$
\mathsf{Ext}: q^{n}\times q^{d}\rightarrow q^{m}
$
such that for any source $X$ with $\sfH_{\infty}(X)\geq r$, we have
\[
\SD((\mathsf{Ext}(X, \mathsf{Seed}), \mathsf{Seed}), (U, \mathsf{Seed}))\leq \delta 
\]
with the $ \mathsf{Seed}$  uniformly distributed over $\bF_q^d$.

A function $\ext: q^{n}\rightarrow q^{m}$ is a (seedless) $(n, m, r, \delta)$-extractor if for any source $X$ with $\sfH_{\infty}(X)\geq r$, the distribution $\ext(X)$ satisfies  $\SD(\ext(X), U)\leq \delta$.
\end{definition}

\remove{
The average min-entropy and average-case strong extractor measures the case that the adversary has side information $Y$ about source $X$. Let the average min-entropy be $\bH_{\infty}(X|Y)=-\log \mathsf{E}_{y\leftarrow Y}(\max_x\bPr(X=x|Y=y))$.
\begin{definition}\cite{DORS08}
Let the source $\bH_{\infty}(X|Y)\geq r$. A (seedless) average-case $(n, m, r, \delta)$-strong extractor is a function,
\[
{\mathsf{Ext}}: q^{n}\rightarrow q^{m}
\]
such that $\SD(({\mathsf{Ext}}(X), Y), (U, Y))\leq \delta$.
\end{definition}
}

%There is a construction of 
A seedless extractor  can be constructed from Reed-Solomon (RS) codes \cite{CDS12}.
%, which is easy to construct and extracts uniform randomness.
The construction works only for a  restricted class of  sources, known as {\em symbol-fixing sources}.

\begin{definition} 
An $(n, m)$ symbol-fixing source is % the distribution of
a tuple of independent random variables ${\bf X} = (X_1, \cdots, X_n)$, defined over a set $\Omega$, 
such that $m$ of the variables take values uniformly and independently from $\Omega$,
and the remaining variables have fixed values.

\end{definition}

We show a construction of a seedless $(n, m, m\log q, 0)$-extractor from RS-codes. Let $q\geq n+m$. Consider an $(n,m)$  symbol-fixing source
 ${\bf X}=(X_1, \cdots, X_n)\in \bF^n_q$ with $\sfH_\infty(X)\geq m\log q$. The extraction has two steps:
\begin{enumerate}
\item %The extractor first c
Construct a polynomial $f(x) \in \bF_q[X]$ of degree $\leq n -1$, such that $f(i) =x_i$ for $i=0, \cdots, n-1$.

\item %Then, the extractor e
Evaluate the polynomial at $i=\{n,\cdots, n+m-1\}$. That is,
\[
\mathsf{Ext}({\bf x})=(f(n),f(n+1),\cdots ,f(n+m-1))
\]
\end{enumerate}

\subsubsection{Limited-View Adversary Code}

Limited-view adversary codes provide reliable communication over an \rw-AWTP.
%%%
\remove{
A model of adversarial channel, called Limited-View Adversary Channel (LVAC), was introduced in \cite{SW13}. In this model the adversary is computationally unlimited but their access to the channel (codeword) is limited as follows: for a codeword of length $n$, the adversary can adaptively choose $\rho_rn$ components to "see" and $\rho_wn$ components to modify, and the modification is by "adding" to the codeword an error vector of weight at most $\rho_wn$.  Limited-view adversary code (LV adversary code) provides reliable communication  against an adversary who has partially view of communication channel, and use this view to corrupt the sent codeword 
}
\cite{SW13,WSN15}. 
%The adversary in limited-view adversary channel is defined  same as the one in AWTP channel, so the LV adversary code can also be used to achieve reliable communication over AWTP channel.
 
\begin{definition}\label{def_awtpcode}
An $(n, k, \delta)$-Limited-View Adversary Code $($or $(n, k, \delta)$-LV adversary code$)$ for a $(\rho_r, \rho_w)$-AWTP channel,
is a code of  dimension $k$ and length $n$.
Encoding and decoding algorithms are $(\LVACenc, \LVACdec)$.
% such that t
The probability  that the receiver output a  message $m'\neq m$,   %with probability 
is bounded by $\delta$. That is for any $m\in\cM$, and adversary's observation $z$ we have,
\[
\bPr (\LVACdec(\LVACenc(m)+\Adv(z))\neq m)\leq \delta.
\] 
\end{definition}
 
%We will use the 
LV adversary codes provide % transmit messages reliably
reliable communication  over AWTP channels.  Previous constructions
%The previous construction of LV adversary code
  achieves capacity $1 - \rho_w$ \cite{WSN15,SW131}, but with  the condition that  $\rho_r + \rho_w <1$.  
In Appendix \ref{ap_lvcode} we give  a simple construction of LV adversary code with low rate of communication, but get rid of the restriction $\rho_r + \rho_w <1$ (Appendix \ref{ap_lvcode}).  Security of this construction is given by the following theorem.

\begin{theorem}\label{lvac}
The LV adversary code has rate $R_\lv = 1 + \rho - \rho_r - 2\rho_w $ over $(\rho_r, \rho_w, \rho)$-AWTP channel.
The probability of error for a length $n$ code is % except with error probability
 $\delta \leq \frac{un}{q}$.   The computation is polynomial in code length $n$.
\end{theorem}

\subsection{Interactive Key Agreement Protocol with Strong Reliability}% over AWTP channel}
\label{sec_sk_constr}

We introduce a three round % key agreement 
interactive protocol with strong reliability,  over AWTP channel.
The idea behind  %of key agreement 
the protocol is as follows.
% that i
In the first  transmission round, Alice  sends a sequence of  randomly selected   components to Bob.
The adversary reads over a  set  $S^r_{1}$, and adds errors on  a $S^w_{1}$.
Bob receives a vector that is partially corrupted and partially leaked to the adversary.
In the second round, Bob generates a key for each component, and uses a MAC algorithm to construct a 
tag for each component, using its ow attached key.
\remove{
 Since the adversary can read a fraction of communication on  sets $S^r_{1}$, and add errors on $S^w_{1}$, the word that Bob received may leak information of randomness to adversary. The randomness also have adversarial errors. In order to 
generate common randomness between Alice and Bob, in second round, 
Bob generates the key and tag pairs to authenticate the randomness received in previous round. 
}
The keys and tag pairs are sent to Alice using an LV-code and so are received correctly by Alice who will
use   the received key and tag pair to check the correctness
of the  $i^{th}$ component.
%who checks  to check if the randomness components that Bob receives are same as the her. The key and tag pairs are encoded by LV adversary code, and so reliably transmitted to Alice. 
In the key derivation step, Alice and Bob use a randomness extractor to generate a shared key %the secure key 
from their shared   randomness which is partially leaked to the adversary.

The construction of secure key agreement protocol uses universal hash function, seedless random extractor, and LV adversary code. 
Let $n_1$, $n_2$, $n_3$ be the length of the protocol messages  sent by Alice and Bob in each round.
The   total
communication  length % of the of protocol 
is $n=n_1+n_2+n_3$. 
Let the protocol be over $\bF_q^u$.   
We use: 1). the $\frac{u}{q}$-$\Delta$ universal $(q, q^{u-1}, q)$-hash family; 2). the seedless $(un_1,\ell,\ell \log q,0)$-extractor; 3). the $(n, k,  \delta)$-LV adversary code with  $k = \frac{2n_1}{u}$,  $n = \frac{k}{R_\lv}$, $\delta \leq \frac{un_1}{q}$ over alphabet $\Sigma = \bF_q^{u}$.

\begin{theorem}
The key agreement protocol has rate $R_\sk=1-\rho-\xi$ when $\rho_r + 2\rho_w < 1 + \rho$, over  %transmission alphabet for the 
 AWTP channel. The alphabet  size is $|\Sigma| = \cO(q^{\frac{1}{\xi}})$. The key agreement protocol is perfectly secure and the decoding error is bounded by $\delta \leq \xi$. The number of round is three. The computation complexity is $\mathcal{O}((n\log q)^2)$. 
\end{theorem}

The secure key agreement protocol is given  in Figure \ref{FigKeyAgreementAWTP}.

\begin{center}
{\bf Fig \ref{FigKeyAgreementAWTP}.} Secure Key Agreement Protocol over AWTP channel
\end{center} 
\begin{framed} \label{FigKeyAgreementAWTP}
\begin{enumerate}

\item R1: Alice $\overset{\mathsf{AWTP}}{\longrightarrow}$ 
Bob. For each $i\in n_1$, Alice chooses a vector $r_{i}$  that is 
uniformly distributed over $\bF_q^{u-1}$, and $\beta_{i}$ 
over $\bF_q$. Alice sends  over the  forward AWTP channel to Bob,  the codeword $c_1=(c_{1,1}, \cdots, c_{1,n_1})$ 
where $c_{1,i}=(r_{i}, \beta_{i})$.

\item R2: Bob $\overset{\mathsf{AWTP}}{\longrightarrow}$ Alice. Bob receives $x_1=(x_{1,1}, \cdots, x_{1,n_1})$ with $x_{1,i}=(r'_{i},\beta'_{i})$ f%rom Alice. Bob
and generates 
a vector of random values  $(\alpha_{1}, \cdots, \alpha_{n_1})$ over $\bF_q$. 
Bob generates ${\bf t}=(t_{1}, \cdots, t_{n_1})$ over $\bF_q$ such that, 

\bequation
t_{i}=\mathsf{MAC}(r'_{i}, \alpha_{i})+\beta'_{i}\mod q, \mbox{\;\;  for\;\;} i=1,\cdots, n_1.
\eequation

Bob encodes $(\alpha_{1}, \cdots, \alpha_{n_1},t_{1}, \cdots, t_{n_1})$ into LV adversary code $d_2$ over $\bF_q^u$. 
Bob sends the codeword $d_2$ over 
backward AWTP channel to Alice.

\item R3: Alice $\overset{\mathsf{AWTP}}{\longrightarrow}$ 
Bob. Alice receives $y_2$, and decode into $(\alpha_{1}, \cdots, \alpha_{n_1},t_{1}, \cdots, t_{n_1})$ using LV adversary code decoding algorithm. For each $i=1 \cdots n_1$, Alice checks if,

\bequation
t_{i} \overset{?}{=} \mathsf{MAC}(r_{i}, \alpha_{i})+\beta_{i}\mod q.
\eequation

Alice generates a binary vector $(v_{1}, \cdots, v_{n_1})$ where   $v_{i}=1$ if $(r_{i}, \beta_{i})$ pass the authentication test, and $v_{i}=0$ if not. Alice encodes $(v_{1}, \cdots, v_{n_1})$ into an LV adversary code $c_3$,  and sends it over the forward AWTP channel to Bob.

\item Key Derivation. Alice and Bob use  a key derivation algorithm to generate secret key. 

\begin{itemize}
\item Alice generates a vector $(s_{1},  
\cdots, s_{n_1})$ with $s_i=r_{i}$ if $v_{i}=1$, $s_{i}=0$ if $v_{i}=0$ for $i,=1 \cdots, n_1$. Alice generates a key $k_A$ using randomness extractor,
\begin{equation}
k_A=\ext(s_{1}, \cdots, s_{n_1})
\end{equation}
\end{itemize}

\begin{itemize}
\item Bob receives $x_3$ from Alice, and decodes it  into $(v_{1}, \cdots, v_{n_1})$ using LV adversary code decoding algorithm. Bob generates a vector $(s'_{1},  
\cdots, s'_{n_1})$ with $s'_i=r'_{i}$ if $v_{i}=1$, $s'_{i}=0$ if $v_{i}=0$ for $i,=1, \cdots, n_1$.  Finally Bob uses the 
extractor to generate  a security key,
\bequation
k_B=\ext(s'_{1} 
\cdots, s'_{n_1})
\eequation
\end{itemize}

\end{enumerate} 
\label{fig_keyagreeprotocol}

\end{framed}

\noindent{\bf Secrecy and Reliability}

Let the secret key $k=k_A$. We  %only need to
will  show that the probability that Alice and Bob output % wrong 
different keys is bound by $\delta$;  that is $\bPr(K = K_A \neq K_B)\leq \delta$.
Moreover  and the distribution of the secret key given the adversary's observation, is %the same as uniformly distribution
uniform, that is $\SD(P_{K|Z}, U) = \SD(P_{K_A|Z}, U) = 0$.

\begin{lemma}
The probability that Alice and Bob do not output the same key % output the wrong key
  is bounded by $\delta \leq  \frac{un}{q}$ if $\rho_r + 2\rho_w \leq 1 + \rho$.
\end{lemma}

\begin{proof}

First we consider the case $\rho_r + 2\rho_w \geq 1 + \rho$. If $\rho_r + 2\rho_w \geq 1 + \rho$, we can not use the LV adversary code to transmit messages $(\alpha_1, \cdots, \alpha_{n_1}, t_1, \cdots, t_{n_1})$ and $(v_1, \cdots, v_{n_1})$ reliably to the other party, respectively. This is because the rate of LV adversary code $R_\lv = 1 + \rho_w - \rho_r - 2\rho_w$.  If $\rho_r + 2\rho_w \geq 1 + \rho$, the rate of LV adversary code $R_\lv \leq 0$. Alice and Bob can not receive $(\alpha_1, \cdots, \alpha_{n_1}, t_1, \cdots, t_{n_1})$ and $(v_1, \cdots, v_{n_1})$ except with negligible error. So the Alice and Bob can not generate secret key using key agreement protocol with negligible error probability.

Then we consider the case $\rho_r + 2 \rho_w \leq 1 + \rho$.  
The secret key generated by Alice not equal to Bob happens in two cases:

\begin{enumerate}

\item Alice and Bob decode the correct message $(\alpha_{1}, \cdots, \alpha_{n_1},t_{1}, \cdots, t_{n_1})$, and $(v_{1}, \cdots, v_{n_1})$ from $y_2$ and $x_3$, respectively, 
using LV adversary code decoding algorithm, except with probability at most $\delta_1 \leq \frac{u(n_2 + n_3)}{q}$.

This is because $(\alpha_{1}, \cdots, \alpha_{n_1},t_{1}, \cdots, t_{n_1})$, and $(v_{1}, \cdots, v_{n_1})$ are encoded by LV adversary code. Since $\rho_r + 2 \rho_w \leq 1 + \rho$, the message $(\alpha_{1}, \cdots, \alpha_{n_1},t_{1}, \cdots, t_{n_1})$, and $(v_{1}, \cdots, v_{n_1})$ can be encoded by LV adversary code with rate $R_\lv > 0$. 
From Theorem \ref{lvac}, the probability that receiver (Alice or Bob) does not output the correct messages $(\alpha_{1}, \cdots, \alpha_{n_1},t_{1}, \cdots, t_{n_1})$, and $(v_{1}, \cdots, v_{n_1})$ are bounded by $\frac{un_2}{q}$ and $\frac{un_3}{q}$, respectively. So both parties outputs correct message from $y_2$ and $x_3$, except with probability at most $\frac{u(n_2 + n_3)}{q}$.

\item Given Alice and Bob share the same $(\alpha_1, \cdots, \alpha_{n_1}, t_1, \cdots, t_{n_1})$ and $(v_1, \cdots, v_{n_1})$, the two parties will generate common randomness $(s_{1}, \cdots, s_{n_1}) = ({s}'_{1}, \cdots, {s}'_{n_1})$, except with probability at most $\delta_2 \leq \frac{un_1}{q}$. 

% if Alice and Bob have the same vectors $(\alpha_{1,f}, \cdots, \alpha_{n_1,f},t_{1,f}, \cdots, t_{n_1,f})$, $(\alpha_{1,b}, \cdots, \alpha_{n_1,b},t_{1,b}, \cdots, t_{n_1,b})$, $(v_{1,f}, \cdots, v_{n_1,f})$, and $(v_{1,b}, \cdots, v_{n_1,b})$.

This is from,  
\begin{equation}\label{eq_sads_1}
\begin{split}
&\bPr((s_{1}, \cdots, s_{n_1})\neq ({s}'_{1}, \cdots, {s}'_{n_1}))\\
&\leq \sum_{i=1}^{n_1}\bPr({s}_i\neq {s}'_i)\\
&= \sum_{i=1}^{n_1}\bPr({s}_i\neq {s}'_i, v_{i}=1)  \\
&\leq \sum_{i=1}^{n_1}\bPr({r}_{i}\neq {r}'_{i} \mbox{ and } \mathsf{MAC}(r_{i},\alpha_{i})-\mathsf{MAC}(r'_{i},\alpha_{i}) = \beta_{i}'-\beta_{i} ) \\
&\leq \frac{un_1}{q}
\end{split}
\end{equation}

\end{enumerate}

Since the secret key $k_A$ and $k_B$ are extracted from randomness $(s_1, \cdots, s_{n_1})$ and $(s'_1, \cdots, s'_{n_1})$, the probability that Alice and Bob generate same secret key such that $k_A =k_B$, is bounded by,  

\begin{equation}
\begin{split}
& 1- \delta  = \bPr(K = K_A = K_B ) \\
& = \bPr \Big( \LVACdec(y_2) = (\alpha_{1}, \cdots, \alpha_{n_1},t_{1}, \cdots, t_{n_1}) \mbox{ and } \LVACdec(x_3) = ( v_{1}, \cdots, v_{n_1} ) \\
&\qquad \mbox{ and } (s_1, \cdots, s_{n_1}) = (s'_1, \cdots, s'_{n_1}) \Big) \\
& = \bPr \Big( \LVACdec(y_2) = (\alpha_{1}, \cdots, \alpha_{n_1},t_{1}, \cdots, t_{n_1}) \mbox{ and } \LVACdec(x_3) = ( v_{1}, \cdots, v_{n_1}) \Big) \\
&\qquad  \bPr \Big( (s_1, \cdots, s_{n_1}) = (s'_1, \cdots, s'_{n_1}) \;|\; \LVACdec(y_2) = (\alpha_{1}, \cdots, \alpha_{n_1},t_{1}, \cdots, t_{n_1}) \\
&\qquad \mbox{ and } \LVACdec(x_3) = ( v_{1}, \cdots, v_{n_1}) \Big) \\
& \geq (1 - \delta_1)(1 - \delta_2) \geq (1 - \frac{u(n_2 + n_3)}{q})(1 - \frac{un_1}{q})\\
& \geq 1 - \frac{u(n_2 + n_3)}{q} - \frac{un_1}{q} = 1 - \frac{un}{q}
\end{split}
\end{equation}

So it implies the reliability of key agreement protocol is bounded by $\delta \leq \frac{un}{q}$.

\end{proof}

\begin{lemma}
The key agreement protocol has perfectly secrecy if $\ell \leq (u-1)(1-\rho)n_1 $
\end{lemma}

\begin{proof}

To show the prefect security of key agreement protocol, we assume the adversary reads on the last $\rho_rn_1$ fraction of codeword. The general adversary attacking can be proved similarly.

First, the vector $(r_1, \cdots, r_{(1 - \rho_r)n_1})$ is perfectly secure for any adversary's observation $z$.
Since adversary reads $\rho_r$ fraction of codeword in first round, and read the message encoded by LV adversary code in the second and third, the adversary's observation is no more than the following set of components,
\bequation
Z=\Big\{ r_{(1-\rho_r)n_1+1}\cdots r_{n_1}, \beta_{(1-\rho_r)n_1+1}\cdots \beta_{n_1},\alpha_{1}, \cdots, \alpha_{n_1},t_{1}, \cdots, t_{n_1}, v_{1}, \cdots, v_{n_1} \Big\}
\eequation
For the set of components $\Big\{ r_1, \cdots, r_{(1 - \rho_r)n_1} \Big\}$, it is perfectly secure for any adversary's observation. It implies the vector $(r_1, \cdots, r_{(1 - \rho_r)n_1})$ has min-entropy at least $\ell \log q$.

Second, since vector $(s_1, \cdots, s_{n_1})$ is generated from $(r_1, \cdots, r_{n_1})$, which has min-entropy at least $\ell \log q$, it implies $(s_1, \cdots, s_{n_1})$ also has min-entropy at least $\ell \log q$.

Finally, since the $k_A$ is generated from $(s_1, \cdots, s_{n_1})$ using randomness extractor, and $(s_1, \cdots, s_{n_1})$ is $(n_1, \ell)$ symbol-fixing source, it implies  the secret key $\SD(K_A | Z) = \SD(U)$. 
 
\qed
\end{proof}

\noindent{\bf Rate of Key Agreement Protocol}

\begin{lemma}
The rate of key agreement protocol approaches $R_{\sk}=1-\rho$.
\end{lemma}

\begin{proof}
For a small $\xi \geq 0$, let the parameter be chosen as $u \geq \frac{1}{\xi} + \frac{4}{\xi R_\lv}$, $q \geq 2u n^2$, $\ell = (u - 1)(1 - \rho) n_1$, $n_0 \geq u$, and $\Sigma = \bF_q^u$. Let $R_\sk = 1 - \rho$. 
For any $n \geq n_0$, the rate of secure key agreement protocol family is given by,

\begin{equation}
\begin{split}
\frac{\log|\cK|}{n\log |\Sigma|} &= \frac{\ell \log q}{un \log q} \\
&= \frac{(u - 1)(1 - \rho) n_1 \log q}{(n_1 + n_2 + n_3)u\log q}\\
&= \frac{(u - 1)(1 - \rho) n_1 \log q}{(n_1 + \frac{2n_1}{u R_\lv} + \frac{2n_1}{u R_\lv})u\log q}\\
&= \frac{ u - 1 }{u + \frac{4 }{ R_\lv}}(1 - \rho)\\
&\geq 1 - \rho -\xi = R_\sk - \xi
\end{split}
\end{equation}

The decoding error probability is bounded by,

\bequation
\delta \leq \frac{un}{q} \leq \frac{1}{2n} \leq \xi
\eequation

From Definition \ref{def_keyagreefamily}, the rate of secure key agreement is $R_\sk = 1 - \rho$.

\end{proof}

\subsection{An SKA Protocol with Strong Reliability over AWTP-PD Channel}

We introduce the key agreement protocol with strong reliability over AWTP-PD channel. Both 
AWTP channel and pubic discussion is over  alphabet $
\Sigma=\bF_q^u$, where $q$ is a prime, and $u$ is an 
integer. The key agreement protocol has three rounds, uses AWTP 
channel once  and teh public discussion channel twice.

The construction of key agreement protocol is similar to the key agreement protocol over AWTP channel (Fig \ref{FigKeyAgreementAWTP}). Since the communication is over PD channel,
 after the first round of %key agreement
 the  protocol, Bob can directly transmit $(\alpha_1, \cdots, \alpha_{n_1}, t_1, \cdots, t_{n_1})$ to Alice in the second round, and Alice can also directly transmit $(v_1, \cdots, v_{n_1})$ to Bob in the third round, without using LV adversary code. 
 The
  difference between key agreement protocol over AWTP-PD channel and key agreement protocol over interactive AWTP channel is 
 that   messages are  directly transmitted in the second and third round of  the protocol.
 This means that the condition %The key agreement protocol over AWTP-PD channel also have no need on condition of LV adversary code, that is
  $\rho_r + 2\rho_w < 1 + \rho$ that was imposed by the LV-code will not be required.
  %, since the message are transmit over PD channel.

\begin{theorem}
The key agreement protocol has rate $R_\sk=1-\rho-\xi$ over  transmission alphabet for the AWTP channel is of size $|\Sigma| = \cO(q^{\frac{1}{\xi}})$, 
 The computation complexity is $\mathcal{O}((n\log q)^2)$. 
\end{theorem}

Proof is in the full version of the paper. It is omitted because f space.

\subsection{An SKA Protocol   with weak relaibility}%l with Weak Reliability}

%In this section, we consider the secure key agreement protocol with
We consider  weak reliability. % for AWTP channel or over AWTP-PD channel. 
The key agreement protocol is  one round.
% Since Alice and Bob does not need to communicate each other after the first round transmission over AWTP channel, we can use the same construction for key agreement protocol with weak reliability over AWTP channel and over AWTP-PD channel.

The construction uses % of secure key agreement protocol with weak reliability uses 
AMD codes and randomness extractors.  The proof is in the full version of the paper.

\begin{theorem}
The key agreement protocol in \ref{fig4} has one round over  AWTP channel, and achieves rate $R_\sk = 1 - \rho_r$.
The alphabet size  is $|\Sigma| = \cO(q^{\frac{1}{\xi}})$. The protocol has   polynomial time  computation. 

\end{theorem}

\begin{center} \label{fig4}
{\bf Fig. \ref{KeyAgreementWeakReliability}} Secure Key Agreement Protocol with Weak Reliability. 
\end{center}

\begin{framed}\label{KeyAgreementWeakReliability}

Alice does the following:
\begin{enumerate}

\item Chooses a vector $s = (s_1, \cdots, s_n)$, that is uniformly distributed over $\bF_q^{u - 2}$.

\item Chooses a vector $(r_1, \cdots, r_n)$ that is uniformly distributed over $\bF_q$, and 
generates $(t_1, \cdots, t_n)$ using AMD code  (Eq \ref{EqAMDFunction}),
\begin{equation}
t_i = f(s, r_i) \mod q
\end{equation}

\item Sends the codeword $c = (c_1, \cdots, c_n)$ over $\bF_q^u$ with $c_i = (s_i, r_i, t_i)$, to Bob over AWTP channel. 

\item Alice generate $k_A$ using randomness extractor.
\begin{equation}
k_A=\ext(s_{1}, \cdots, s_{n_1})
\end{equation}

\end{enumerate}

Bob does the following:
\begin{enumerate}

\item Receives the word $x = (x_1, \cdots, x_n)$ with $x_i = (s'_i, r'_i, t'_i)$ and checks if $x$ is tampered by Eve by checking:
\begin{equation}
t'_i \overset{?}{=} f(s', r'_i)
\end{equation}

\item Output $\perp$ if $x$ is tampered by Eve. Otherwise, Bob generates $k_B$ using randomness extractor. 
\begin{equation}
k_B=\ext(s'_{1}, \cdots, s'_{n_1})
\end{equation}

\end{enumerate}

\end{framed}

The above protocol shows that under weak reliability, very efficient key agreement protocols can be constructed.

\section{Concluding remarks}
We motivated and defined a new setting for key agreement protocols where the adversary partially controls the communication 
channel, and interaction over this channel is  the only resource of the adversary.
Previous works had considered the cases that the channel was fully authenticated, or fully corrupted. In such 
a setting channel by itself cannot be the only resource for establishing a secret shared key: in the former
case no secrecy for the can be provided, and in the latter no guarantee on communication.
All protocols in these settings assume prior dependent variables as communicants' resource for establishing a shared key.
In our setting, the limited control of the adversary makes the channel a resource for extracting a shared key.
We formalized the model, derived the secret key rate bounds, and gave constructions that achieve the bounds.

There are numerous open questions that follow form this work. First and foremost, construction of protocols for small alphabets.
Our constructions although have constant size  alphabet, but the alphabet size depends on how close  the rate of the protocol 
is to the upper bound.  The alphabet size determines granularity of the physical layer adversaries.  In network setting, each component of a protocol message (codeword) will be sent over a path and so larger size  alphabets could be acceptable. In wireless communication however, 
the alphabet size must be  reduced.

Secondly, we defined leakage and corruption as constant ratios of the transmitted word.
One can consider other measures of leakage and corruption to limit the adversary's power.

Thirdly, we motivated the use of physical layer properties of communication systems for providing security  against massive surveillance systems.
We showed partially controlled physical environments 
 can be used to establish shared secret keys between two participants.
Designing other cryptographic primitives that use partial access of the adversary to the physical resources of a system is an
interesting direction for future work.

Finally,
the three round protocol in Section \ref{sec_sk_constr}   has the requirement $\rho_r+2\rho_w\leq 1+\rho$ among parameters. 
Achieving the bound without this requirement, and finding the minimum  number of rounds for  protocols 
with similar property (achieve the upper bound), are open problems.

\remove{
\begin{enumerate}
\item The first open question is how to construct key agreement protocol over AWTP without the requirement $\rho_r+\rho_w<1$. 
\item The second open question is how to construct key agreement protocol over general adversarial wiretap channel. 
\end{enumerate}
  }

\bibliographystyle{abbrv}
\bibliography{1.bib}

\remove{
\bibitem{Mau93} Maurer,  ``Secrey Key Agreement by Public Discussion from Common Information", \emph{IEEE Trans. on IT}, 1993.
\bibitem{Mau96}   Maurer, The strong secret key rate of discrete random triplets, communications and cryptography.
\bibitem{Mau97} Information Theoretically Secure Secret-Key Agreement by NOT Authenticated Public Discussion?,  \emph{Eurocrypt},1997.
%[MW97] 
\bibitem{MW97}
Ueli M. Maurer and StefanWolf. Privacy amplification secure against active adversaries. In CRYPTO,
pages 307?321, 1997.
\bibitem{MW03}
 Ueli M. Maurer and Stefan Wolf. Secret-key agreement over unauthenticated public channels iii:
Privacy amplification. IEEE Transactions on Information Theory, 49(4):839?851, 2003.
\bibitem{Bra} Bennett, Brassard, Crepeau, Maurer, ``Generalized Privacy Amplification",  \emph{IEEE Trans. on IT}, 1995.
\bibitem{Wyn} Wyner, ``Wiretap channel"
b
\bibitem{CFH15} Eric Chitambar and Benjamin Fortescue and Min-Hsiu Hsieh, ``Distributions Attaining Secret Key at a Rate of the Conditional Mutual Information", CRYPTO, 2015.
\end{thebibliography}
}

\begin{appendix}
%%%%%% PROOF OF LEMMA 1
\section{Proof of Lemma \ref{le_sec}}\label{ap_le_sec} 
\begin{proof}
1). First, we show $\sfI(K; C^{\ell, r} D^{\ell, r}) \leq 4\epsilon  n\log(\frac{|\Sigma|}{\epsilon})$. 

The proof uses Pinsker's Lemma:
\begin{lemma}
Let $P$, $Q$ be probability distributions. Let $ \SD(P , Q)\leq \epsilon$. Then 
\[
\sfH(P ) - \sfH(Q) \leq 2\epsilon \cdot \log(\frac{|P\cup Q|}{\epsilon})
\]
\end{lemma}

From the definition of secrecy of key agreement (Eq. (\ref{eq_sec})), we have, 
\[
\SD(P_{K|Z}, U)\leq \epsilon
\]
From Pinsker's lemma and adversarial reading sets $Z=(C^{\ell, r} D^{\ell, r})$, it implies,
\bequation \label{eq_sec1}
\sfH(U)-\sfH(K|C^{\ell, r} D^{\ell, r}) \leq 2\epsilon\cdot \log(\frac{|\cK|}{\epsilon})
\eequation

Since $U$ is the uniform distribution over $\cK$, it implies, $\sfH(K)\leq \log |\cK|=\sfH(U)$. Since Alice and Bob's randomness $r_A$ and $r_B$ are not correlated, it implies $|\cK|\leq |\Sigma|^{2n}$. So there is,
\begin{equation} \label{eq_sec2}
\begin{split}
\sfI(K; C^{\ell, r} D^{\ell, r})&\leq \sfH(K)-\sfH(K|C^{\ell, r} D^{\ell, r})\\
&\leq\sfH(U)-\sfH(K|C^{\ell, r} D^{\ell, r})\\
&\leq 2\epsilon\cdot \log(\frac{|\cK|}{\epsilon})\\
&\leq 4\epsilon  n\log(\frac{|\Sigma|}{\epsilon})
\end{split}
\end{equation}

2). Second, we show $\log |\cK|-\sfH(K)\leq 2\epsilon\cdot \log(\frac{|\cK|}{\epsilon})$. 

From $\sfH(K|C^{\ell, r} D^{\ell, r}) \leq \sfH(K) \leq \log |\cK|$, $\sfH(U) = \log |\cK|$, and (Eq. \ref{eq_sec1} and \ref{eq_sec2}), it implies,
\bequation
\log |\cK| - \sfH(K) \leq 4\epsilon n \log(\frac{|\Sigma|}{\epsilon})
\eequation
\qed
\end{proof}

%%%%%%%%%%%%%%%Proof of Lemma 2
\section{Proof of Lemma \ref{le_rel}}\label{ap_le_rel}

\begin{proof}
From the definition of reliability (Eq. (\ref{pf_bound1})), the probability that Bob outputs the wrong key is bounded as follow,
\bequation\label{pf_rel1}
\Pr(K_A\neq K)\leq \delta \mbox{\;\;\;\;\;\; and \;\;\;\;\;}  \Pr(K_B\neq K)\leq \delta
\eequation
From Fano's lemma and (Eq. \ref{pf_rel1}), it implies, 
\bequation 
\sfH(K|C^\ell Y^\ell)=\sfH(K|K_A)\leq \sfH(\delta)+\delta\log |\cK|
\eequation
and,
\bequation 
\sfH(K|D^\ell X^\ell)\leq \sfH(K|K_B)\leq \sfH(\delta)+\delta\log |\cK|
\eequation
Since $C^{\ell,a}=X^{\ell,a},C^{\ell,d}=X^{\ell,d},D^{\ell,a}=Y^{\ell,a},D^{\ell,d}=Y^{\ell,d}$, it implies,

\bequation \label{pf_rel2}
\sfH(K|C^{\ell,a}C^{\ell,b}C^{\ell,c}C^{\ell,d} D^{\ell,a}Y^{\ell,b}Y^{\ell,c}D^{\ell,d})\leq \sfH(\delta)+\delta\log |\cK|
\eequation
and, 
\bequation \label{pf_rel3}
\sfH(K|D^{\ell,a}D^{\ell,b}D^{\ell,c}D^{\ell,d} C^{\ell,a}X^{\ell,b}X^{\ell,c}C^{\ell,d})\leq \sfH(\delta)+\delta\log |\cK|
\eequation

From,
\begin{equation}
\begin{split}
&\sfH(K D^{\ell,b}D^{\ell,c} X^{\ell,b}X^{\ell,c} | D^{\ell,a}D^{\ell,d}C^{\ell,a}C^{\ell,d})\\
&=\sfH(K | D^{\ell,b}D^{\ell,c} X^{\ell,b}X^{\ell,c} D^{\ell,a}D^{\ell,d}C^{\ell,a}C^{\ell,d})+\sfH(D^{\ell,b}D^{\ell,c} X^{\ell,b}X^{\ell,c} | D^{\ell,a}D^{\ell,d}C^{\ell,a}C^{\ell,d})\\
&=\sfH(D^{\ell,b}D^{\ell,c} X^{\ell,b}X^{\ell,c} | D^{\ell,a}D^{\ell,d}C^{\ell,a}C^{\ell,d} K ) +\sfH(K| D^{\ell,a}D^{\ell,d}C^{\ell,a}C^{\ell,d}) 
\end{split}
\end{equation}

it implies,
\begin{equation} \label{pf_rel4}
\begin{split}
&\sfH(K| D^{\ell,a}D^{\ell,d}C^{\ell,a}C^{\ell,d}) \\
&\leq  \sfH(D^{\ell,b}D^{\ell,c} X^{\ell,b}X^{\ell,c} | D^{\ell,a}D^{\ell,d}C^{\ell,a}C^{\ell,d}) - \sfH(D^{\ell,b}D^{\ell,c} X^{\ell,b}X^{\ell,c} | D^{\ell,a}D^{\ell,d}C^{\ell,a}C^{\ell,d} K ) \\
&\;\;\;\;+\sfH(\delta)+\delta\log |\cK|
\end{split}
\end{equation}

From,
\begin{equation}
\begin{split}
&\sfH(D^{\ell,b}D^{\ell,c} X^{\ell,b}X^{\ell,c} | D^{\ell,a}D^{\ell,d}C^{\ell,a}C^{\ell,d}C^{\ell,b}C^{\ell,c}Y^{\ell,b}Y^{\ell,c})\\
&=\sfH(K | D^{\ell,a}D^{\ell,d}C^{\ell,a}C^{\ell,d}C^{\ell,b}C^{\ell,c}Y^{\ell,b}Y^{\ell,c}) \\
&\qquad + \sfH(D^{\ell,b}D^{\ell,c} X^{\ell,b}X^{\ell,c} | D^{\ell,a}D^{\ell,d}C^{\ell,a}C^{\ell,d}C^{\ell,b}C^{\ell,c}Y^{\ell,b}Y^{\ell,c} K)\\
&\qquad - \sfH(K | D^{\ell,b}D^{\ell,c} X^{\ell,b}X^{\ell,c} D^{\ell,a}D^{\ell,d}C^{\ell,a}C^{\ell,d}C^{\ell,b}C^{\ell,c}Y^{\ell,b}Y^{\ell,c})\\
&\leq \sfH(K | D^{\ell,a}D^{\ell,d}C^{\ell,a}C^{\ell,d}C^{\ell,b}C^{\ell,c}Y^{\ell,b}Y^{\ell,c}) \\
&\qquad + \sfH(D^{\ell,b}D^{\ell,c} X^{\ell,b}X^{\ell,c} | D^{\ell,a}D^{\ell,d}C^{\ell,a}C^{\ell,d}C^{\ell,b}C^{\ell,c}Y^{\ell,b}Y^{\ell,c} K)\\
&\overset{(1)}{\leq} \sfH(\delta)+\delta\log |\cK| +  \sfH(D^{\ell,b}D^{\ell,c} X^{\ell,b}X^{\ell,c} | D^{\ell,a}D^{\ell,d}C^{\ell,a}C^{\ell,d} K)
\end{split}
\end{equation}
Here (1) is from (Eq. \ref{pf_rel2}), and, 
\begin{equation}
\begin{split}
&\sfH(D^{\ell,b}D^{\ell,c} X^{\ell,b}X^{\ell,c} | D^{\ell,a}D^{\ell,d}C^{\ell,a}C^{\ell,d}C^{\ell,b}C^{\ell,c}Y^{\ell,b}Y^{\ell,c} K)\\
& \qquad\qquad\qquad\qquad\qquad   \leq \sfH(D^{\ell,b}D^{\ell,c} X^{\ell,b} X^{\ell,c} | D^{\ell,a}D^{\ell,d}C^{\ell,a}C^{\ell,d} K),
\end{split}
\end{equation}
It implies,
\begin{equation} \label{pf_rel5}
\begin{split}
& \sfH(D^{\ell,b}D^{\ell,c} X^{\ell,b}X^{\ell,c} | D^{\ell,a}D^{\ell,d}C^{\ell,a}C^{\ell,d} C^{\ell,b}C^{\ell,c}Y^{\ell,b}Y^{\ell,c} )\\
& \qquad \leq \sfH(D^{\ell,b}D^{\ell,c} X^{\ell,b}X^{\ell,c} | D^{\ell,a}D^{\ell,d}C^{\ell,a}C^{\ell,d} K) + \sfH(\delta)+\delta\log |\cK| 
\end{split}
\end{equation}

From (Eq. \ref{pf_rel4} and \ref{pf_rel5}), it implies,
\begin{equation} \label{pf_rel6}
\begin{split}
&\sfH(K| D^{\ell,a}D^{\ell,d}C^{\ell,a}C^{\ell,d}) \\
&\leq   \sfH(D^{\ell,b}D^{\ell,c} X^{\ell,b}X^{\ell,c} | D^{\ell,a}D^{\ell,d}C^{\ell,a}C^{\ell,d})\\
& - \sfH(D^{\ell,b}D^{\ell,c} X^{\ell,b}X^{\ell,c} | D^{\ell,a}D^{\ell,d}C^{\ell,a}C^{\ell,d} C^{\ell,b}C^{\ell,c}Y^{\ell,b}Y^{\ell,c} ) + 2\sfH(\delta)+ 2\delta\log |\cK|\\
&= \sfI(D^{\ell,b}D^{\ell,c} X^{\ell,b}X^{\ell,c}; C^{\ell,b}C^{\ell,c}Y^{\ell,b}Y^{\ell,c} | D^{\ell,a}D^{\ell,d}C^{\ell,a}C^{\ell,d}) + 2\sfH(\delta)+2\delta\log |\cK|
\end{split}
\end{equation}

From $adv_1$, since $(X^{\ell,b},X^{\ell,c})=(C^{\ell,b},C^{\ell,c})+(E^{\ell,b},E^{\ell,c})$ and $(Y^{\ell,b},Y^{\ell,c})=(D^{\ell,b},D^{\ell,c})+({E'}^{\ell,b},{E'}^{\ell,c})$, it implies,

\begin{equation} \label{pf_rel7}
\sfI(D^{\ell,b}D^{\ell,c} X^{\ell,b}X^{\ell,c}; C^{\ell,b}C^{\ell,c}Y^{\ell,b}Y^{\ell,c} | D^{\ell,a}D^{\ell,d}C^{\ell,a}C^{\ell,d}) = 0
\end{equation}
From (Eq. \ref{pf_rel6} and \ref{pf_rel7}) it implies,
\[
\sfH(K| D^{\ell,a}D^{\ell,d}C^{\ell,a}C^{\ell,d})\leq 2\sfH(\delta)+2\delta\log |\cK|
\]
\qed
\end{proof}

\end{appendix}

%%%%%%%%lvcode%%%%%%%%

\section{LV Adversary Code}\label{ap_lvcode}

The construction of LV adversary code use the algebraic manipulation detection code (AMD code).

\subsubsection{LV adversary code Construction}

The construction of LV adversary code achieves reliable communication over LV adversary channel (or AWTP channel), even with the condition of $\rho_r + \rho_w \geq 1$. 
The idea of LV adversary code construction is that sender first encode the message using Reed-Solomon code. Then for each components of Reed-Solomon code, sender uses AMD code to authentication each component. When the receiver receives the LV adversary code, for each components the receiver uses AMD code verification algorithm to check whether the components has been tampered by adversary or not. For the components of LV adversary code that adversary only write, there is high chance to detect the error. For the rest of components of codeword, the receiver use the Reed-Solomon code decoding algorithm to output the correct message. 

Let the message of LV adversary code has length $\ell$. The length of LV adversary code is $n$. The LV adversary code is over $\bF_q^u$. To construct LV adversary code, we use $(\ell, n)$-Reed-Solomon code over $\bF_q^{u-2}$, and $(\bF_q^{u-2}, \bF_q^{u}, \frac{d+1}{q})$-AMD code.
The encoding algorithm is in Figure \ref{fig_lvencoding1}, and decoding algorithm is in Figure \ref{fig_lvencoding1}.

\begin{center}
{\bf Fig.\ref{fig_lvencoding1}.} LV adversary code Encoding algorithm 
\end{center}
\vspace{-4mm}
\begin{framed} 
\begin{enumerate}

\item Step1: For message $m$, the sender encodes the message into Reed-Solomon code $c_\rs$. 

\item Step2: For each component $c_i$ for $i=1, \cdots, n$, the sender encodes $c_i$ into AMD code $(c_i, r_i, t_i)$. 

\end{enumerate}
\label{fig_lvencoding1}
\end{framed}

The LV adversary can read on set $S_r$ and add error on set $S_w$. The receiver receives corrupted word $y$.

\begin{center}
{\bf Fig.\ref{fig_lvdecoding2}.} LV adversary code Decoding algorithm 
\end{center}
\vspace{-4mm}
\begin{framed} 
\begin{enumerate}

\item Step1: For each component $y_i = (c_i', r_i', t_i')$ for $i=1, \cdots, n$, the receiver uses AMD code verification algorithm to check if the AMD code is valid, that is $t_i' \overset{?}{=}f(c_i', r_i')$. 

\item Step2: The receiver discard the error components. For the rest components passing the AMD code verification algorithm, the receiver uses Reed-Solomon code decoding algorithm to output the message $m$. 

\end{enumerate}
 \label{fig_lvdecoding2}
\end{framed}

We show the rate of the LV adversary code.

\begin{theorem}
The LV adversary code achieves rate $R = 1 - \rho_r - 2\rho_w + \rho$ over $(\rho_r, \rho_w, \rho)$-AWTP channel, except with error probability $\delta \leq \frac{un}{q}$, in $Poly(n)$.
\end{theorem}

\begin{proof}
We denote the components $|S_r \cap S_w| = \rho_0 n$. Since $|S_r \cap S_w| + |S_r \cup S_w| = |S_r| + |S_w|$, it implies, $\rho_0 + \rho = \rho_r + \rho_w$.

The components of codeword can be divided into four categories: not corrupted, read only, read and write, write only.

For the write only components $(c_i, r_i, t_i)$, since the adversary does not know the AMD code $(c_i, r_i, t_i)$, the probability that adversary tampered AMD code $(c'_i, r'_i, t'_i)$ passes verification is bounded by $\frac{u}{q}$. Since there are at most $n$ write only components, the probability of any writing only components pass AMD code verification algorithm is bounded by $\frac{un}{q}$. 

For the rest components including not corrupted, read only, read and write components, the length of codeword is $n' = n - (\rho_w - \rho_0)n$. Since the length of error  is $\rho_0 n$, the receiver can uniquely output the correct message if the length of message $\ell \leq n'- 2\rho_0 n$. So the rate of codeword is bounded as follow,

\begin{equation}
\begin{split}
R &\leq \frac{\ell}{n} = \frac{ n' - 2\rho_0 n }{n} = 1 - (\rho_w - \rho_0) - 2\rho_0 \\
&= 1 - \rho_w - \rho_0 = 1 +\rho - \rho_r - 2\rho_w 
\end{split}
\end{equation}

\end{proof}

\section{Relation Between Upper Bound of Key Agreement}\label{ApRelationUpBound}

We show the relation between the upper bound of key agreement \cite{MAU97} in which Alice and Bob generate key using public discussion from the shared triple variable with distribution $\bPr_{XYZ}$, and key agreement over AWTP-PD channel.

\begin{lemma}
The rate of key agreement is $R_\sk \leq \frac{1}{n \log |\Sigma|} \min( \sfI(X^n; Y^n), \sfI(X^n; Y^n | Z^n) )$.
 The relation between mutual information entropy and reading and writing parameters $(\rho_r, \rho_w)$ of adversary wiretap channel are $\sfI(X^n ; Y^n | Z^n) = (1 - \rho)n \log |\Sigma|$ and $\sfI( X^n; Y^n) \leq (1 - \rho_w)n \log |\Sigma|$. So the rate of key agreement protocol is bounded by $R_\sk \leq 1 - \rho$. 
\end{lemma}

\begin{proof}

We assume that Eve reads on the $S_r$ components of $(X_1, \cdots, X_n)$ with $|S_r| = \rho_r n$, and add random error on the $S_w$ components of codeword with $|S_w| = \rho_w n$. The components that adversary either read or write is $S = S_r \cup S_w$ with $|S| = \rho n$.

First, from Theorem 4 \cite{MAU93}, the rate of key agreement protocol $R_\sk$ is bounded by,
\[
R_\sk \leq \frac{1}{n \log |\Sigma|} \min( \sfI(X^n; Y^n), \sfI(X^n; Y^n | Z^n) )
\]

Second, we show that $\sfI(X^n;  Y^n) \leq (1 - \rho_w)n \log |\Sigma|$. 

We have,
\[
\sfH( Y^n | X^n) = \sfH( E^n | X^n ) = \sfH( E^n) = \rho_w n \log |\Sigma|
\]
So, 
\[
\begin{split}
& \sfI ( X^n; Y^n) = \sfH( Y^n ) - \sfH( Y^n|X^n) \leq n\log |\Sigma| - \rho_w n \log |\Sigma| \\
&= ( n - \rho_w n ) \log |\Sigma| 
\end{split}
\]

Third, we show that $\sfI( X^n; Y^n | Z^n) \leq (1 - \rho)n \log |\Sigma|$.  Let $\overline{Z}^n$ be the random variable that adversary does not read. Since $Z^n$ is equal to $X^n$ on set $S_r$, and zero on $[n]/S_r$, it implies $\overline{Z}^n$ is equal to $X^n$ on set $[n]/S_r$, and zero on $S_r$. So there is $X = Z + \overline{Z}$. We have,

\[
\begin{split}
\sfI(X^n ; Y^n |  Z^n ) & = \sfH( X^n | Z^n ) - \sfH( X^n | Y^n, Z^n) \\
& = \sfH(X^n, Z^n ) - \sfH(Z^n) - ( \sfH( X^n, Y^n, Z^n) - \sfH( Y^n, Z^n ) )\\
& =  \sfH( Y^n | Z^n ) - \sfH( Y^n | X^n )\\
& = \sfH( X^n + E^n | Z^n) - \rho_wn \log |\Sigma|\\
& = \sfH( \overline{Z}^n + Z^n + E^n | Z^n) - \rho_wn \log |\Sigma|\\
& \leq \sfH( \overline{Z}^n + E^n ) - \rho_wn \log |\Sigma|\\
& \leq (n - (\rho - \rho_w)n )\log|\Sigma| - \rho_wn \log |\Sigma|\\
& = (1 - \rho)n \log|\Sigma| 
\end{split}
\]
\qed
\end{proof}

\end{document}